\documentclass[lettersize,journal]{IEEEtran}

 \pdfoutput=1
 
\usepackage [noadjust]{cite}
\usepackage{amsmath,amssymb,amsfonts}

\usepackage{balance}
\usepackage{graphicx}
\usepackage{caption}

\usepackage{bm}  

\usepackage{enumerate}
\usepackage{stfloats}
\usepackage{cases}
\usepackage{epstopdf}
\usepackage{soul,color} 
\usepackage{hyperref}
\usepackage[percent]{overpic}
\usepackage{tabularx}
\usepackage[table]{xcolor}
\usepackage{multirow}
\usepackage{booktabs}  
\usepackage{tabu} 

\usepackage{algorithm}
\usepackage{algpseudocode}%
\usepackage{textcomp}

\newcommand{\eps}{\varepsilon}

\newcommand{\argmax}{\operatornamewithlimits{argmax}}

\newcommand{\hbold}{\boldsymbol{h}}
\newcommand{\Hbold}{\boldsymbol{H}}
\newcommand{\Ebold}{\boldsymbol{E}}

\newcommand{\sbold}{\boldsymbol{s}}
\newcommand{\kbold}{\boldsymbol{k}}

\newcommand{\gbold}{\boldsymbol{g}}
\newcommand{\pbold}{\boldsymbol{p}}
\newcommand{\Ibold}{\boldsymbol{I}}
\newcommand{\Abold}{\boldsymbol{A}}

\newcommand{\abold}{\boldsymbol{a}}
\newcommand{\bbar}{\overline{b}}
\newcommand{\bbold}{\boldsymbol{b}}
\newcommand{\bboldbar}{\overline{\boldsymbol{b}}}

\newcommand{\aboldhat}{\widehat{\boldsymbol{a}}}

\newcommand{\rbold}{\boldsymbol{r}}
\newcommand{\rboldf}{\boldsymbol{r}^{\mathrm{f}}}
\newcommand{\rf}{r^{\mathrm{f}}}
\newcommand{\rboldbar}{\overline{\boldsymbol{r}}}
\newcommand{\rbar}{\overline{r}}
\newcommand{\thetabar}{\overline{\theta}}

\newcommand{\dF}{d^{\mathrm{F}}}
\newcommand{\dFMax}{\overline{d}^{\mathrm{F}}}
\newcommand{\dFzero}{d^{\mathrm{F0}}}
\newcommand{\dFzeroMax}{\overline{d}^{\mathrm{F0}}}

\newcommand{\dNzeroMax}{\overline{d}^{\mathrm{N0}}}
\newcommand{\dNR}{d^{\mathrm{NR}}}

\newcommand{\bboldbarMR}{\overline{\boldsymbol{b}}^{\mathrm{MR}}}

\newtheorem{theorem}{Theorem}
\newtheorem {corollary}{Corollary}

\newtheorem {property}{Property}
\newtheorem {remark}{Remark}
\newtheorem {definition}{Definition}

\begin{document}

	\raggedbottom
	\allowdisplaybreaks

     \title{A Study on Characterization of Near-Field Sub-Regions For Phased-Array Antennas
  \thanks{This research was supported by Business Finland via project 6GBridge - Local 6G (Grant Number: 8002/31/2022), and Research Council of Finland, 6G Flagship Programme (Grant Number: 346208).}}
	\author{Mehdi Monemi, \textit{Member}, IEEE, Sirous Bahrami, \textit{Member}, IEEE, Mehdi Rasti, \textit{Senior Member}, IEEE, Matti Latva-aho, \textit{Fellow}, IEEE 
		\thanks{
  A part of this work has been presented and published
in IEEE International Symposium on Personal, Indoor and Mobile Radio
Communications (IEEE PIMRC) 2024, Valencia, Spain, and a part has been accepted for presentation and publication in IEEE Global Communications Conference (IEEE GLOBECOME) 2024, Cape Town, South Africa.
}
\thanks{
  Mehdi Monemi, Mehdi Rasti, and Matti Latva-aho are with the University of Oulu, 90570 Oulu, Finland (emails: mehdi.monemi@oulu.fi, mehdi.rasti@oulu.fi, matti.latva-aho@oulu.fi).
}
\thanks{
Sirous Bahrami is with the Department of Electrical Engineering, Pohang University of Science and Technology (POSTECH), Pohang 37673, South Korea (e-mail:
bahramis@postech.ac.kr).
}
}

\maketitle	

\begin{abstract}
We characterize three near-field sub-regions for phased array antennas by elaborating on the boundaries {\it Fraunhofer}, {\it radial-focal}, and {\it non-radiating} distances. The {\it Fraunhofer distance} which is the boundary between near and far field has been well studied in the literature on the principal axis (PA) of single-element center-fed antennas, where PA denotes the axis perpendicular to the antenna surface passing from the antenna center. The results are also valid for phased arrays if the PA coincides with the boresight, which is not commonly the case in practice. In this work, we completely characterize the Fraunhofer distance by considering various angles between the PA and the boresight. 
For the {\it radial-focal distance}, below which beamfocusing is feasible in the radial domain, a formal characterization of the corresponding region based on the general model of near-field channels (GNC) is missing in the literature. We investigate this and elaborate that the maximum-ratio-transmission (MRT) beamforming based on the simple uniform spherical wave (USW) channel model results in a radial gap between the achieved and the desired focal points. While the gap vanishes when the array size $N$ becomes sufficiently large, we propose a practical algorithm to remove this gap
in the non-asymptotic case when $N$ is not very large. Finally, the {\it non-radiating} distance, below which the reactive power dominates active power, has been studied in the literature for single-element antennas. We analytically explore this for phased arrays and show how different excitation phases of the antenna array impact it. We also clarify some misconceptions about the non-radiating and Fresnel distances prevailing in the literature. 
\end{abstract}

\section{Introduction}
\label{sec:introduction} 
The sixth generation of wireless technology (6G) promises phenomenal advancements, aiming to deliver speeds hundreds of times faster than current capabilities \cite{9770093}. This ambitious goal is fueled by cutting-edge technologies like extremely large-scale antenna arrays (ELAAs) \cite{10379539} and operation at very high frequencies, including millimeter waves (mmWave) and terahertz (THz) bands \cite{10040913}. However, these advancements lead to consequences like significantly expanding the near-field region of electromagnetic (EM) wave propagation, extending its reach to potentially hundreds of meters 
\cite{cui2022near6G}.
 Signal propagation in the near-field region poses challenges for designing near-field communication systems but also offers more Degrees-of-Freedom (DoFs) and new opportunities for novel applications that are not feasible in the far-field. Some of these challenges and DoFs include locally varying polarization of EM waves, beamforming in the 3D domain, non-stationary antenna element characteristics,  non-sparse channel representation of ELAAs, and holographic properties of antenna arrays enabling holographic communication. These challenges and DoFs begin to appear when entering the near-field region characterized by the Fraunhofer distance.
This paper aims to provide a comprehensive and categorized study of various electromagnetic propagation regions and present a more generalized perspective for the characterization of each region. We focus on three key boundaries\footnote{The {\it Fresnel} distance is another important near-field boundary investigated in many existing works for single-element antennas \cite{selvan2017fraunhofer,balanis2016antenna,capozzoli2024review}. In this paper, however, we do not investigate this due to space limitations. The extensions of the results for the Fresnel distance of phased array antennas can be found in \cite{monemi2024GLOBECOM}.}:
\begin{itemize}
    \item {\it Fraunhofer distance}, which is the transition boundary between the near-field and far-field.
     \item {\it Radial focal distance}, identifying the region inside the near-field wherein beamfocusing in the radial domain is feasible.
    \item {\it Non-radiative distance}, below which the non-radiative reactive power dominates the active power.
\end{itemize}

These distances serve as critical boundaries for understanding signal behavior in the near-field region. The  {\it Fraunhofer} distance has been extensively studied and characterized in the literature for single-element center-fed antennas \cite{selvan2017fraunhofer,balanis2016antenna,capozzoli2024review} 
on the principal axis (PA) of the antenna, where the PA denotes the axis perpendicular to the antenna surface passing from the geometrical center (i.e., feed location) of the antenna. The results are also valid for phased arrays if the PA coincides with the boresight, which is not always the case in practice. In our work, we completely characterize the Fraunhofer distance by considering various angles between the PA and the boresight. 
The {\it radial focal distance} and the corresponding {\it radial focal region} are only applicable to phased arrays in the near-field region; Despite limited exploration, a comprehensive study regarding the formal definition and characterization of this region is also absent in the existing literature. Especially most existing works have analyzed and characterized the beamfocusing in the radial domain based on the uniform spherical wave (USW) channel model in the near-field region (e.g., \cite{alexandropoulos2022near}), and formal characterization based on the more general near-field channel models less investigated, which is to be explored in this paper. 
Finally, the non-radiating region has been analytically studied for a single-element (dipole) antenna, however, this has not been studied for phased arrays under different excitation phases, and besides, there exists some misconceptions in the literature about the discrimination of Fresnel distance and non-radiating distance which will be elaborated in this work.
To address these research gaps, we revisit existing definitions, perform detailed calculations, and provide characterizations for each of the above three boundaries by considering phased array antennas, and wherever possible, we present exact closed-form expressions. The regions characterized in this paper are schematically depicted in Fig. \ref{fig:regions}.  

\subsection{Background and Contributions}
In what follows, we explore the three regions related to the near-field propagation and the corresponding boundaries mentioned above. For each item, we commence by providing its background context. Subsequently, we delve into the specific contributions conducted in this paper. Before starting the literature review, we present some terms and definitions used in the following parts. We define the {\it principle axis} of the antenna as the axis perpendicular to the antenna surface, passing from the geometrical central point of the antenna. The {\it boresight} is the axis corresponding to the main lobe of the antenna. This is identical to the principal axis for single-element center-fed symmetrical antennas, however, for phased arrays, we use this term to represent the axis connecting the geometrical antenna center and the intended observation/transmitting point source. The principal axis and boresight are depicted for single-element and phased array antennas in Figs. \ref{fig:Fraunhofer_single} and \ref{fig:Fraunhofer_revisit} respectively. The {\it on-boresight} and {\it off-boresight} refer to the cases where the principal axis does and does not, respectively, coincide with the boresight of the antenna.     

{\hspace{-14pt} \it 1- \underline{Fraunhofer region}}

{\bf Background:} The Fraunhofer distance (Rayleigh distance) is the distance between the observation point and the antenna below which we have the near-field with spherical propagation, and beyond which the far-field exists wherein the spherical incident wavefront can be approximated as planar. 
It is well-known that a maximum phase delay of $\pi/8$ between a point on the antenna surface and the feed, caused by the curvature of the wavefront, is small enough to approximate the spherical wavefront as planar. This phase delay has been accepted in the literature for defining and calculating the Fraunhofer distance. 
The Fraunhofer distance on the antenna's boresight is obtained as $\dFzeroMax=2D^2/\lambda$, where $D$ is the maximum dimension of the antenna and $\lambda$ is the wavelength. 

Several works have calculated the Fraunhofer distance for a single-element antenna following different schemes.  The most straightforward scheme is to calculate the dominant terms of the binomial or Taylor series expansion approximation of the signal, and then obtain the maximum distance from the antenna where the dominant terms would result in a {\it differential phase delay} no more than $\pi/8$ on different points of the antenna surface \cite{balanis2016antenna,stutzman2012antenna}. As an alternative solution, the authors of \cite{selvan2017fraunhofer} have derived $\dFzeroMax$ by considering a continuous aperture of arbitrary shape, and then they employed scalar diffraction theory for obtaining Fraunhofer and Fresnel regions. The study by \cite{7590187} demonstrates that for single-element thin-wire dipole antennas, the boresight Fraunhofer distance $2D^2/\lambda$ is accurate only for antennas where $D \geq 5\lambda$. 
Noting the inexactness of $2D^2/\lambda$ when $D$ is comparable to $\lambda$, mathematical formulations have been obtained in \cite{9896807} for characterizing a more exact on-boresight Fraunhofer distance for a thin wire of length $D$. In \cite{10541333}, the {\it effective Fraunhofer distance} $\dFzeroMax_{\mathrm{eff}}$ is defined and characterized as the distance where the normalized beamforming gain under the far-field assumption is no less than a value $\eta$. By considering $\eta=95\%$, it has been shown that we have $\dFzeroMax_{\mathrm{eff}}=0.367\dFzeroMax$ on the boresight of the antenna. The authors of \cite{10278686} have shown that for MIMO phased array antennas with maximum dimensions of $D_1$ and $D_2$ corresponding to a transmitter and a receiver, the on-boresight Fraunhofer distance is obtained as $\dFzeroMax=2(D_1+D_2)^2/\lambda$.
To the best of our knowledge, all works in the literature except \cite{hu2023design} have analyzed the Fraunhofer distance on the principal axis of the antenna (e.g., \cite{cui2022near6G,yan2021joint,cui2022near,9978148}). In practice, however, it is of a very small probability that the observation point is exactly placed on the boresight. The authors of \cite{hu2023design}  have partially investigated the fact that for the case of an off-boresight phased array scenario, the Fraunhofer distance is increased compared to the on-boresight scenario, however, a comprehensive and analytical investigation is missing. In particular, a detailed analysis of how the observation angle $\theta$  exactly impacts the Fraunhofer distance, and how the Fraunhofer distance changes softly when moving from on-boresight to off-boresight angles is missing.  

{\bf Contributions:}
As in \cite{monemi2024GLOBECOM}, we delve into an analysis of the well-established derivation of the Fraunhofer distance, which traditionally applies to the boresight of a center-fed antenna model, and show that this procedure does not directly translate to off-boresight scenarios. This specifically holds for the phased array antennas, wherein each antenna element has its feed, and thus there exists not only a single central feed reference point.  As a result, the widely accepted values of the on-boresight Fraunhofer distance for phased arrays, are not valid in the off-boresight scenario. Then we derive a closed-form value for the Fraunhofer distance for phased arrays for $\theta\in[0,\pi]$. More specifically, we define and calculate the {\it Fraunhofer angle} $\theta^{\mathrm{F}}$ and show that for $\theta\in[0,\pi/2-\theta^{\mathrm{F}}] \cup [\pi/2+\theta^{\mathrm{F}},\pi]$ the Fraunhofer distance is $\dF=8D^2\sin^2(\theta)/\lambda$. 
For the case where $\theta$ is softly changed from $\pi/2\pm \theta^{\mathrm{F}}$ toward $\pi/2$ (where $\pi/2$ corresponds to the boresight), $\dF$ softly switches to $2D^2/\lambda$ according to a function whose value is derived in a closed form.
We obtain a tight closed-form approximation for the value of $\theta^{\mathrm{F}}$ and show that it is a small angle, its value being a decreasing function of $D/\lambda$; i.e., 
 for phased array antennas, as the antenna array scales up, this angle sharply tends toward zero. We will elaborate on how the Fraunhofer distance for phased arrays is increased about 4 times when debating from the boresight scenario. The approximation becomes tight when studying ELAAs, wherein a very slight deviation of the (user equipment) UE location from the on-boresight scenario (as is commonly the case in practice) extends the Fraunhofer region exactly 4 times.
 The schematic presentation of the characterized Fraunhofer distance $\dF$ for the on-boresight and off-boresight scenarios is depicted in Fig. \ref{fig:regions} per radiation angle $\theta$ in dashed blue line.

{\it \hspace{-14pt} 2- \underline{Radial-domain beamfocusing region}}

{\bf Background:}
The far-field 2D angular domain beamforming has been extensively investigated in the literature for phased array antennas \cite{7342886,8371237,9140420}. The angular-domain beamforming of phased array antennas on the near-field region has also been studied in the literature \cite{10035952,9473882,10129111}, however, one of the most distinguishing and unique features of the near-field is the capability of beamfocusing in the radial domain, which if employed with extremely large-scale arrays can lead to spot-like near-field 3D beamfocusing \cite{monemi2023towards,monemi20236GFresnel}.
 For the asymptotic case where the number of array elements $N$ is extremely large, the authors of \cite{monemi2023towards} have explored a comprehensive study on different technical and practical aspects of implementing spot beamfocusing (SBF) in the near-field region. In \cite{monemi20236GFresnel}, the authors provide a mathematical framework to prove that spot-like 3D radial and angular beamfocusing corresponds to the maximum power beamformer in the asymptotic case under certain conditions. Then they propose a modular (sub-array based) meta-lens structure to focus the beam at the desired focal point near-field region in a smart manner without requiring channel state information (CSI); This idea has then been further developed in \cite{fallah2024near} by implementing transfer-learning for training subarrays in a more adaptive manner. For the non-asymptotic case where the number of array elements is not extremely large, several works have explored the practical and theoretical aspects of near-field beamfocusing in the radial domain. In this regard, to characterize the region wherein near-field beamfocusing in the radial domain is feasible, the idea of {\it depth of focus} has first been defined for continuous apertures in \cite{1137900} and then characterized and applied to phased array antennas \cite{6740832,bjornson2021primer,Liu2023near}. For phased array antennas having small-size antenna patch elements and using Fresnel approximation, the authors of \cite{bjornson2021primer} have proved that the 3dB depth of focus is possible at distances lower than ${\dF}/{10}$. 
For the case of a uniform linear array antenna, an approximate closed-form solution for 3dB depth of focus is calculated in \cite{Liu2023near} based on the uniform spherical wave (USW) channel model for near-field region. 

{\bf Contributions:} 
The depth of focus and radial focal properties investigated in existing works (e.g., \cite{Liu2023near}) are analyzed based on the simple USW model for near-field channels which ignores the variation of the near-field channels of different antenna array elements in the distance domain. We show that this leads to inaccuracies in characterizing the focal point in the radial domain. As in \cite{monemi2024PIMRC}, we express a formal definition of the radial focal distance and present the characterization of the region wherein the constitution of a radial focal point is feasible. Given a desired focal point (DFP) in the radial domain, the beamformer vector corresponding to the highest signal amplitude potentially results in the realization of a radial focal point at some other point between the antenna and the DFP. This leads to a {\it radial focal gap}  between the DFP and the achieved focal point (AFP). By employing ELAAs, we show that this gap tends to zero in the asymptotic case for most practical ELAA structures, considering the non-uniform spherical model (NUSW) for near-field channels. For the case of non-asymptotic cases where the number of array elements is not too large, we provide a practical algorithm for resolving the gap, and realizing the focal point at the exact desired location point.

{\it \hspace{-14pt} 3- \underline{Non-radiating region}}

{\bf Background:} Transmit antennas emit both active and reactive powers. The active power is associated with the traveling wave, while the reactive power is associated with the capacitive or inductive fields. The distance close to the antenna where the active and reactive power levels are equal is termed the {\it non-radiating distance} denoted by $\dNR$. The region $r<\dNR$ corresponds to the non-radiative near-field region where the reactive power dominates, and the region $r>\dNR$ constitutes the radiative region which includes both the radiative near-field ($\dNR<r<\dF$), and the far-field ($r>d^{\mathrm{F}}$) regions as depicted in Fig. \ref{fig:regions}. There has been no analytic study of $\dNR$ for $N$-element phased array antennas, not even for the simple case of $N=1$. To the best of our knowledge, the only theoretical investigation of $\dNR$ in the literature pertains to an infinitesimal dipole for which it is proven that $\dNR=\lambda/2\pi$ \cite{balanis2016antenna,stutzman2012antenna}. 
The non-radiating near-field region has been analyzed in \cite{8954758} through numerical results for three types of antennas: dipole, loop, and Yagi–Uda antennas. More specifically, the extreme value of $\dNR=\lambda/2\pi$, previously derived in the literature through theoretical analysis for infinitesimal dipole antennas, has been verified through numerical results to apply to all three aforementioned types, considering infinitesimal dimensions.

{\bf Contributions:} Similar to the Fraunhofer distance, the Fresnel distance has been characterized based on analyzing a specific phase delay relating to the curvature of the wave for different points on the antenna. The value of Fresnel distance has been derived in several works as $\dNzeroMax=0.62\sqrt{{D^3}/{\lambda}}$ \cite{balanis2016antenna, selvan2017fraunhofer}.
Many related works in the literature assume that the radiative region commences at the Fresnel distance (e.g. \cite{zhang2022beam,gowda2016wireless, you2023near}); we show that this assumption is not correct, as these two regions serve distinct purposes. The former is characterized through an analysis of active and reactive power, whereas the latter is defined based on an examination of phase delay on the aperture surface.
Initiating from the electrical potential function and then applying Maxwell’s equations, we derive exact mathematical expressions for the computation of the complex transmit power relating to an $N$-element dipole array of arbitrary size. By taking the active and reactive parts, we calculate the non-radiating distance $\dNR$ for a dipole array. We evaluate $\dNR$ for various numbers of array elements $N$, various dimensions of each antenna element $D^{\mathrm{s}}$, and various excitation phases for the antenna array elements. Our analysis reveals that, for all stated scenarios, $\dNR$ is consistently lower than half a wavelength, which is significantly smaller than the Fresnel distance of the phased arrays. Besides, it is seen through numerical results that a non-co-phase excitation vector results in a slightly larger value of $\dNR$ compared to implementing a completely co-phase excitation vector. 
\begin{figure}
    \centering
 \includegraphics[width=244pt]{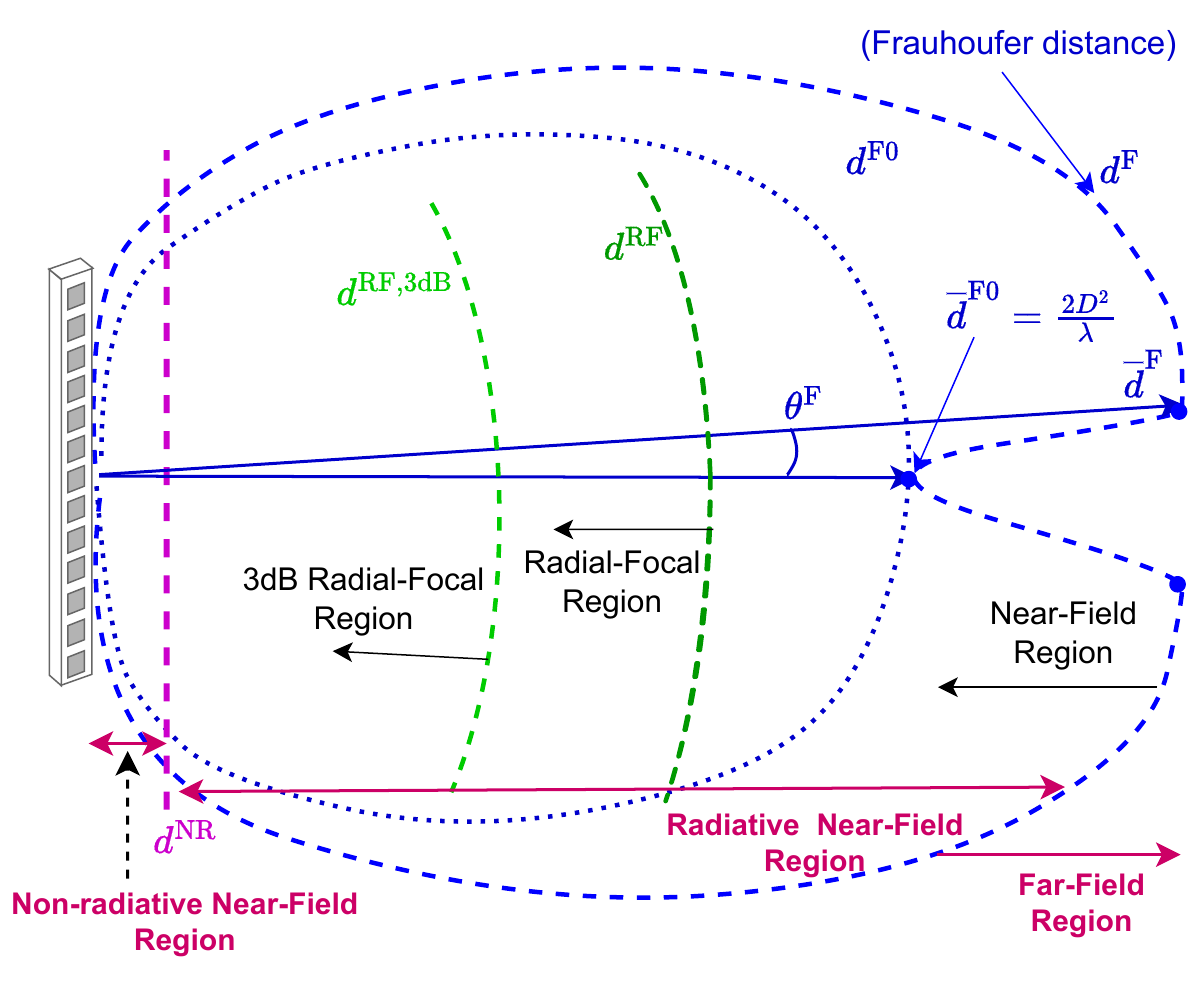}
    \caption{Various characterized near-field sub-regions for a phased array antenna (dashed lines), versus a center-fed single-element antenna having the same dimension (dotted lines).} 

    \label{fig:regions}
\end{figure}

\subsection{Organization} 
The remainder of this paper is organized as follows. In Section \ref{sec:system_model} we present the system model. The characterization of the Fraunhofer region, radial focal region, and non-radiating region is presented in Sections \ref{sec:Fraunhofer_Fresnel}, \ref{sec:radial} and \ref{sec:nonradiating} respectively. 
Finally, Section \ref{sec:conclusions} concludes the paper.    

\subsection{Notations} Throughout this paper, for any matrix $\Abold$, $\Abold_{m,n}$, $\Abold^{\mathrm{T}}$, $\Abold^*$, and $\Abold^\mathrm{H}$, denote the $(m,n)$-th entry, transpose,
conjugate, and conjugate transpose
respectively. Similarly, for each vector $\abold$, $\abold_{n}$, $\abold^{\mathrm{T}}$, $\abold^*$, $\abold^\mathrm{H}$, $\|\abold\|$
denote the $n$-th entry, transpose,
conjugate, conjugate transpose, and Euclidean norm
respectively.

\section{System Model}
\label{sec:system_model}
This section introduces the system model used to characterize the Fraunhofer and the radial focal region, presented in Sections \ref{sec:Fraunhofer_Fresnel} and \ref{sec:radial}, respectively. Studying the non-radiating region requires a different system model that accounts for the exact expressions of the electric and magnetic fields based on Maxwell’s equations, which will be presented in Section \ref{sec:nonradiating}.
 \begin{figure}
    \centering
    \begin{tabular}{c} \includegraphics[width=234pt]{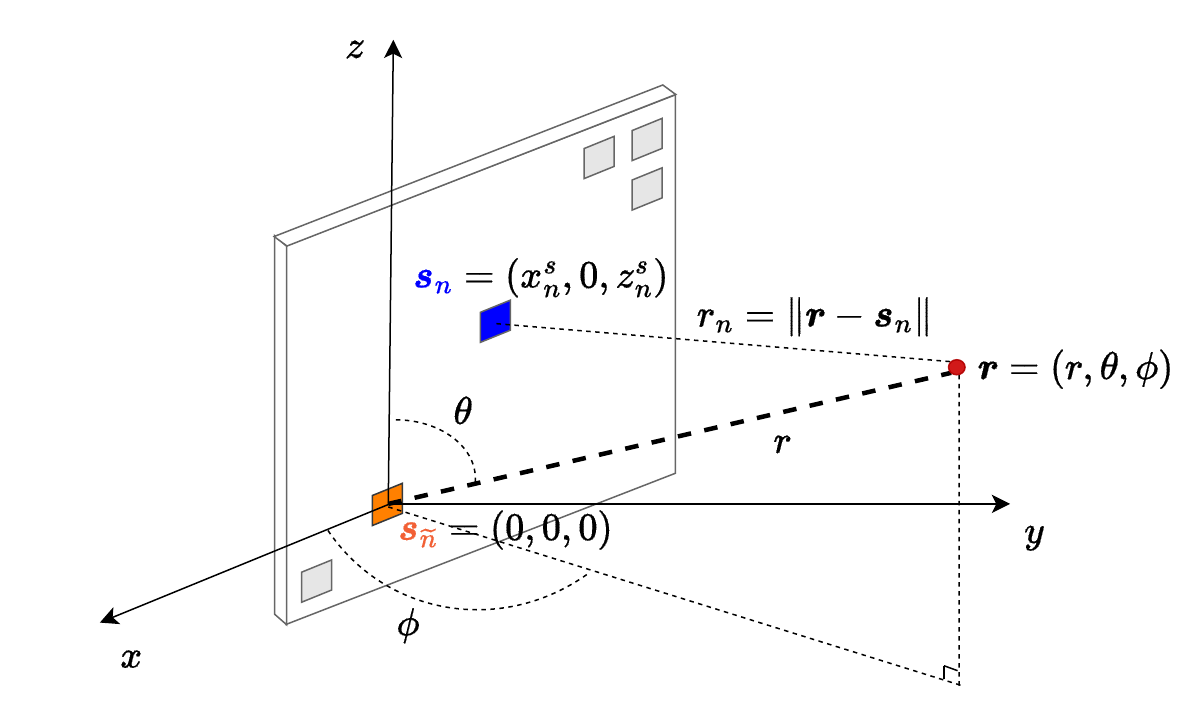}
    \\
    (a): UPA system model
    \\

    \includegraphics[width=234pt]       {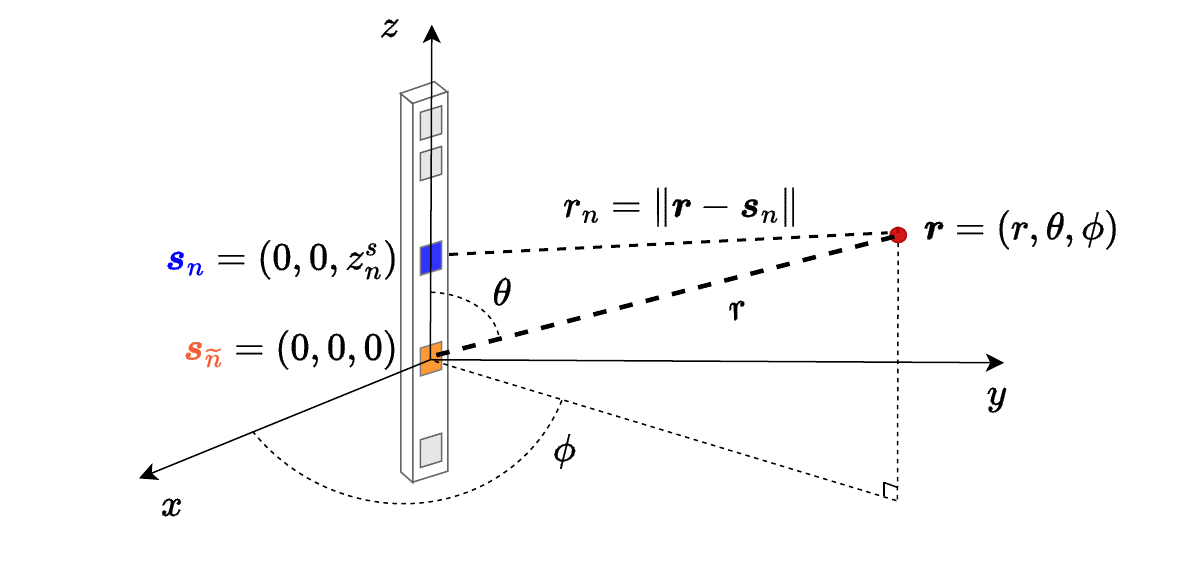}
    \\
    (b): ULA system model
    \end{tabular}
    \caption{The system model for UPA and ULA antennas.} 

    \label{fig:system_model}
\end{figure}
Consider a phased array transmit/receive antenna consisting of $N$ radiating elements as seen in Fig. \ref{fig:system_model}, where the antenna can be a uniform planar array (UPA) located on the $xz$-plane (Fig. \ref{fig:system_model}-a) or a uniform linear array (ULA) located on the $z$-axis (Fig. \ref{fig:system_model}-b). The location of each antenna element $n$ is denoted by $\sbold_n=[x_n^{\mathrm{s}},y_n^{\mathrm{s}},z_n^{\mathrm{s}}]^{\mathrm{T}}$ in the Cartesian coordinate system, and one of the antenna elements indexed as $\widetilde{n}$ is located at the origin (i.e. ${\sbold}_{\widetilde{n}}=[0,0,0]^{\mathrm{T}}$) and considered as the {\it reference antenna element}. The reference point (which can be a transmitting or receiving point) is located at $\rbold=[r,\theta,\phi]^\mathrm{T}$ in the Spherical system, and the distance from point $\rbold$ to the antenna element $n$ is denoted by $r_n$. 
We can write $r_n$ in terms of $r$ and $\sbold_n$ as follows:
\begin{align}
\label{eq:rn_versus_r}
    r_n=\| \rbold-\sbold_n \|=\| r \kbold^{\mathrm{T}}(\theta,\phi)-\sbold_n\|\notag
    \\
    ^=\sqrt{r^2-2r\kbold^{\mathrm{T}}(\theta,\phi)\sbold_n+\|\sbold_n\|^2}
\end{align}
where $\kbold=[\sin\theta \cos \phi, \sin \theta \sin \phi, \cos\theta  ]^{\mathrm{T}}$. 
 If we assume that the antenna operates in transmit mode, the electric/magnetic signal at the receiving point $\rbold$ corresponding to the beamformer $\bbold$ is obtained according to the {\it general near-field channel} (GNC) model in Line-of-Sight (LoS) scenario as follows\cite{Liu2023near}:  
\begin{align}
\label{eq:E126}
    \!\!y(\rbold,\bbold)&= \big[\underbrace{\boldsymbol{l}(\rbold) \odot\gbold(\rbold) \odot \abold(\rbold)}_{\hbold(\rbold)}\big]^{\mathrm{T}} \bbold x
    \notag
    \\
    &= \sum_{n=1}^{N}\sqrt{{ {G_1(\rbold,\sbold_n)G_2(\rbold,\sbold_n)}}/{4\pi r_n^2}}e^{j \left(-\frac{2\pi}{\lambda}   r_n+\beta_n\right) }
\end{align}
where $x$ is the single stream input signal which is assumed to be unity, $\bbold=[e^{j\beta_n}]_{N\times 1}$ is the beamforming vector in which $\beta_n$ is the excitation phase of the $n$'th antenna element,  $\abold(\rbold)=[e^{-j\frac{2\pi}{\lambda}r_n}]_{N\times 1}$ is the array response, $\boldsymbol{l}(\rbold)=[{1}/{\sqrt{4\pi r_n^2}}]_{N\times 1}$ accounts for the free-space path loss, and finally $\gbold_n(\rbold)=[\sqrt{ G_1(\rbold,\sbold_n)G_2(\rbold,\sbold_n)}]_{N\times 1}$ is antenna elements gain vector in which $G_1(\rbold,\sbold_n)$ and $G_2(\rbold,\sbold_n)$ are effective aperture loss and polarization loss corresponding to antenna element $n$ respectively. 
Therefore, the channel gain corresponding to the observation point $\rbold$ based on the {\it general near-field channel} (GNC) model is written as 
\begin{align}
    \mathrm{GNC\ model:}\ h_n=\frac{1}{\sqrt{4\pi}}\times \frac{g_n(\rbold)}{r_n}e^{-j\frac{2\pi}{\lambda}r_n}
\end{align}
For the case where all array elements are considered isotropic point sources, we can consider $G_1(\rbold,\sbold_n)=G_2(\rbold,\sbold_n)=1$, which reduces the {\it general model} to the {\it non-uniform spherical wave} (NUSW) model presented as 
\begin{align}
    \mathrm{NUSW\ model:}\ h_n= \frac{1}{\sqrt{4\pi} r_n}  \times e^{-j\frac{2\pi}{\lambda}r_n}
\end{align}
Finally, if we approximate equal path loss for all antenna elements in the distance domain, and consider only the spherical near-field phase variations of each array element, the simplified channel model known as the {\it uniform spherical model} (USW) is expressed as
\begin{align}
    \mathrm{USW\ model:}\ h_n=\frac{1}{\sqrt{4\pi} r} \times e^{-j\frac{2\pi}{\lambda}r_n}
\end{align}

\section{Characterization of Fraunhofer Region for Off-boresight Scenarios}
\label{sec:Fraunhofer_Fresnel}
In what follows, first, we review the definitions and calculations of the on-boresight Fraunhofer distance \cite{balanis2016antenna,selvan2017fraunhofer}. Then, we extend the results to apply for the off-boresight scenarios.

\subsection{On-boresight Fraunhofer distance}
The on-boresight Fraunhofer distance is characterized by studying the wavefront's curvature through analyzing the array elements' arrival phase. Therefore, the USW near-field channel model applies here.
Consider a transmitting point source and a single-element center-fed receiving antenna with diameter $D$ as depicted in Fig. \ref{fig:Fraunhofer_single}-a. Let $r$ and $r'$ be the distance from the transmitting source to the center  (feed point) and some point $\pbold'$ on the antenna respectively, where the distance between the antenna center and point $\pbold'$ is denoted by $d'$, and the angle between the lines connecting respectively the transmitting point and $\pbold'$ to the feed is denoted by $\theta$. From \eqref{eq:rn_versus_r}, $r'$ is written in terms of $r$,  $d'$, and $\theta$ as follows:
\begin{align}
\label{eq:64328999}
    r'=
    \sqrt{r^2+(-2rd'\cos \theta +(d')^2) }
\end{align}
By using the binomial expansion, the phase difference between the signals arrived at $\pbold'$ and the antenna center is obtained as 
\begin{multline}
\label{eq:deltatheta}
    \Delta\theta=\frac{2\pi}{\lambda}\left(r'-r\right)
=
    \underbrace{\frac{-2\pi}{\lambda}\cos \theta d' }_{\Delta\theta_1}
+
\underbrace{\frac{\pi}{\lambda}  \sin^2 \theta  (d')^2 r^{-1}}_{\Delta\theta_2}
+
\\
\underbrace{\frac{\pi}{\lambda} \cos \theta\sin^2 \theta (d')^3 r^{-2}}_{\Delta\theta_3}
    + \ ...
\end{multline}
It is seen that the first term $\Delta \theta_1$ is independent of the distance $r$. 
It is well known that a phase error equal to $\pi/8$  caused by the curvature of the wavefront is small enough to approximate the spherical wavefront as planar. Therefore, the Fraunhofer distance function of the antenna denoted by $\dFzero$ is defined as the boundary limit for which the maximum phase error between some point on the antenna and the feed is equal to $\pi/8$. Noting that $\Delta \theta_2$ is the main distance-dependant contributor term of $\Delta \theta$, we obtain $\dFzero$ as follows:
\begin{align}
\label{eq:Fraunhofer_single_calc}
    \dFzero=r,\ \mathrm{s.t.}\ \Delta\theta_2\bigg|_{ d'=D/2}=\pi/8.
\end{align}
From \eqref{eq:deltatheta}, 
$\dFzero$ is calculated as
\begin{align}
\label{eq:Fraunhofer_single}
    \dFzero={2D^2} \sin^2(\theta)/\lambda
\end{align}
and thus, the maximum {\it on-boresight Fraunhofer distance} corresponding to the boresight of the antenna ($\theta=\pi/2$) denoted by $\overline{d}^{\mathrm{F0}}$ is obtained as 
\begin{align}
    \label{eq:Fraunhofer_single_max}
    \dFzeroMax={2D^2}/{\lambda}
\end{align}
\begin{remark}
    Considering the Taylor series expansion of \eqref{eq:64328999}, one can verify that ignoring the terms $\Delta \theta_i,$ for $i\geq 3$ results in 
    negligible error
    on the upper bound of phase estimation for about $6\%$ relative to the target value $\frac{\pi}{8}$ for a single-element half-wavelength antenna (corresponding to $D=\lambda/2$). For phased arrays, the estimation is generally more tight. For example, considering a $5$-element ULA with inter-element spacing of half a wavelength (i.e., $D=2\lambda$), the upper bound of phase error resulted from ignoring $\Delta \theta_i,$ for $ i \geq 3$ is obtained equal to $0.2\%$. 
\end{remark}


\begin{figure}
    \centering
    \includegraphics[width=244pt]{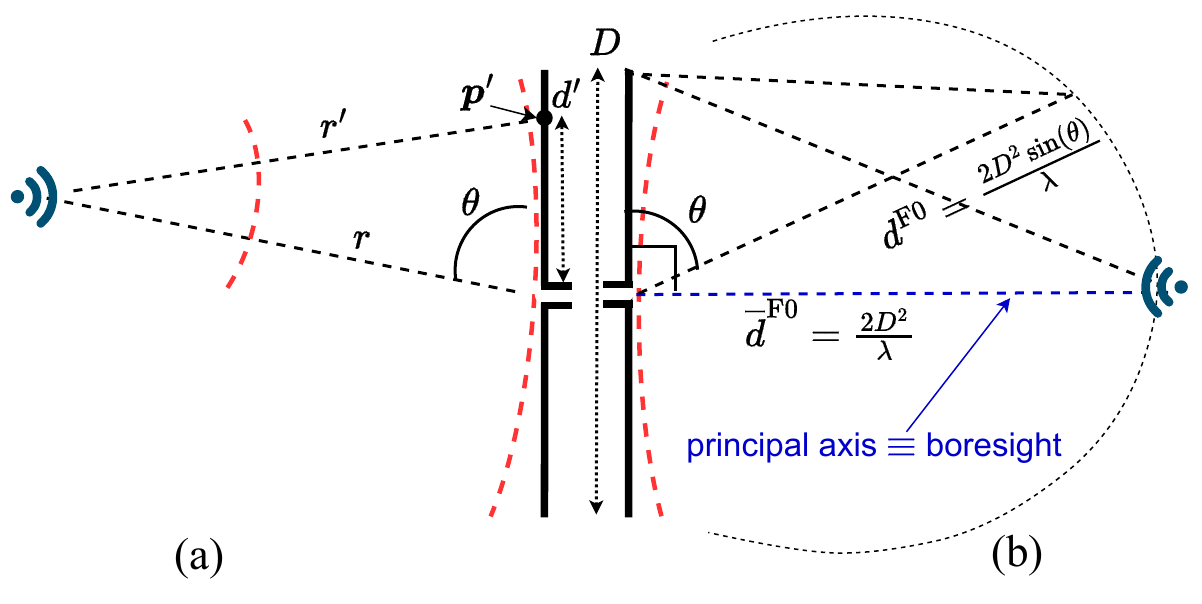}
    \caption{Antenna with diameter $D$. (a): UE located at an arbitrary distance.  (b): UE located at Fraunhofer distance on the boresight.}
    \label{fig:Fraunhofer_single}
\end{figure}




\subsection{Characterizing the off-boresight Fraunhofer region for phased array antennas}


In this part, we elaborate on how the results for on-boresight Fraunhofer distance should be revisited to apply to phased array antennas for off-boresight scenarios. Consider a phased array antenna with the largest dimension of $D$. For simplicity we consider the ULA case, however, the results can be applied to the planar arrays such as UPA antennas as well. In this section, we consider $\frac{D}{\lambda}\geq 0.5$ which is commonly the case by having an inter-element spacing of half a wavelength and $N \geq 2$, or inter-element spacing lower than 0.5 (e.g., as in holographic surfaces) and a higher number of antenna elements.  
As seen in Fig. \ref{fig:Fraunhofer_single}, for a single-element antenna of diameter $D$, a common practice is to calculate the Fraunhofer distance considering the maximum phase error corresponding to the maximum distance of $\pbold'$ to the center of the antenna which is equal to $D/2$. This is not however an exact scheme for calculating the Fraunhofer distance of phased array antennas in the off-boresight case. A phased array antenna consists of multiple elements, each having an independent feed. Therefore, unlike the conventional single-element center-fed antenna model, a more exact approach is to consider the maximum phase delay corresponding to the largest delay difference between any two elements of the array.
\begin{figure}
    \centering
    \includegraphics[width=244pt]{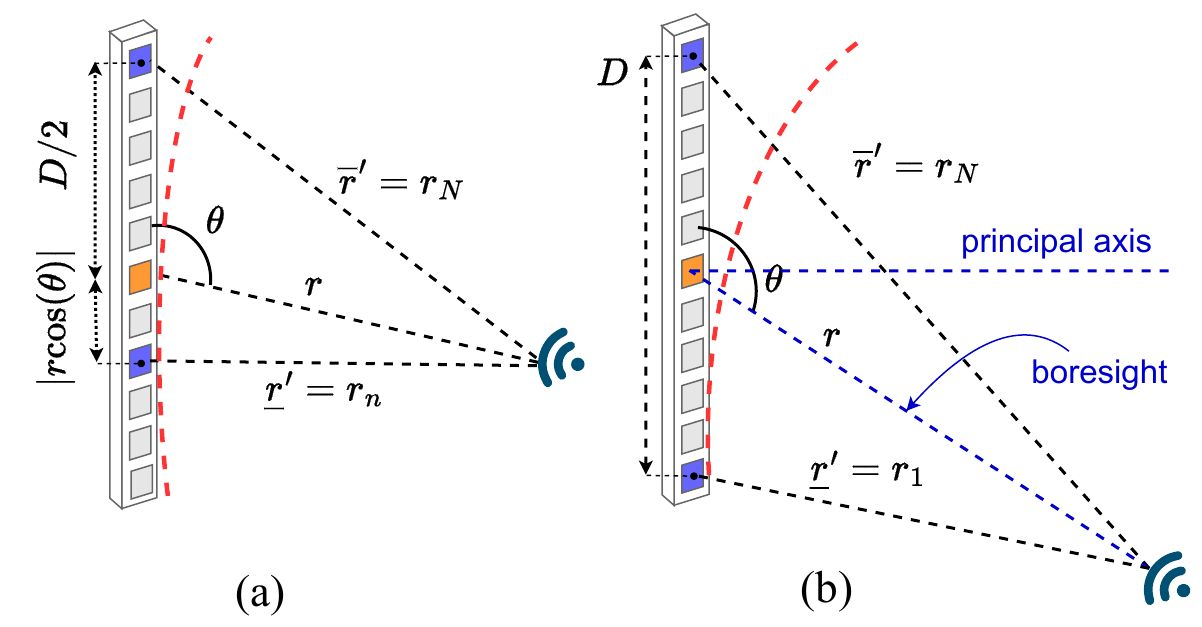}
    \caption{$N$-element ULA antenna with diameter $D$ for two scenarios regarding the position of the transmitter. 
}
    \label{fig:Fraunhofer_revisit}
\end{figure}
Fig. \ref{fig:Fraunhofer_revisit} illustrates an $N$-element ULA antenna with diameter $D$ representing two different cases for calculating the Fraunhofer distance. It is seen that the maximum phase delay is experienced for array elements with distance $\frac{D}{2}+|r\cos(\theta)|$ and $D$ corresponding to Figs. \ref{fig:Fraunhofer_revisit}-a and \ref{fig:Fraunhofer_revisit}-b respectively. {This is different from  the center-fed single-element antenna wherein the maximum phase delay has always been characterized in the literature considering the maximum distance from the geometrical center (i.e. $D/2$)}  \footnote{For the case shown in Fig. \ref{fig:Fraunhofer_revisit}-b, we have considered the approximation that the line perpendicular to the antenna plane (i.e., the line with distance $\underline{r}'$) coincides the center of some antenna element $n$.}

\begin{theorem}
\label{th:dF}
For phased array antennas, the Fraunhofer distance $\dF$ 
is obtained as 
\begin{align} 
\label{eq:dF234}
\dF =
\begin{cases}
     \frac{\lambda}{8}\times\frac{1-\frac{D}{\lambda}F(\theta)-\sqrt{1-\frac{2D}{\lambda}F(\theta)}}{2F(\theta)|\cos(\theta)|},
     & \text{if} \ \frac{\pi}{2}-\theta^{\mathrm{F}} \leq \theta \leq \frac{\pi}{2} +\theta^{\mathrm{F}}
 \\
8D^2\sin^2(\theta)/\lambda
    , & \text{otherwise } 
\end{cases} 
\end{align}
where $F(\theta)=8|\cos(\theta)|\sin^2(\theta)$, and $\theta^{\mathrm{F}}=\pi/2-F^{-1}(\frac{\lambda}{2D})$ in which $F^{-1}(\frac{\lambda}{2D})$ is the value of $\theta$ corresponding to the solution of $F(\theta)=\frac{\lambda}{2D}$ closet to $\pi/2$.
\end{theorem}
\begin{proof}
   Considering the two scenarios depicted in Fig. \ref{fig:Fraunhofer_revisit}-a and Fig. \ref{fig:Fraunhofer_revisit}-b, the Fraunhofer distance is calculated as follows.
\begin{align}
\label{eq:Fraunhofer_array_calc22}
    \dF=r,\ \mathrm{s.t.}\ \Delta\theta_2\bigg|_{d'=D/2 + \Delta D}=\pi/8.
\end{align}
where
\begin{align}
\label{eq:DeltaD}
    \Delta D=\min \{ |r\cos \theta |,D/2 \}.
\end{align}
From \eqref{eq:Fraunhofer_array_calc22} and \eqref{eq:DeltaD}, and considering the term $\Delta\theta_2$ expressed in \eqref{eq:deltatheta}, $\dF$ is found by solving the following equation:
\begin{align}
\label{eq:dF110}
    \frac{2D^2 }{\lambda}\sin^2(\theta)\left( 1+\min\left\{ 1, \frac{2\dF|\cos{\theta}|}{D} \right\}
    \right)^2=\dF
\end{align}
If we consider $2\dF|\cos(\theta)|< D$, after solving \eqref{eq:dF110} and applying the Taylor series, $\dF$ is found as follows.
\begin{multline}
\label{eq:dF1}
    \dF=\frac{\frac{\lambda}{8}-D|\cos(\theta)|\sin^2(\theta) -\frac{\lambda}{8}\sqrt{1-D|\cos(\theta)|\sin^2(\theta)\frac{16}{\lambda}}}
    {2\sin^2{\theta}\cos^2{\theta}}
    \\
    =
    \frac{2D^2 }{\lambda}\sin^2(\theta)
    +
    \frac{16D^3|\cos(\theta)| }{\lambda^2}\sin^4(\theta) + ...
\end{multline}
It is seen in \eqref{eq:dF1} that we have$\dF\geq 2D^2\sin^2(\theta)/\lambda$ and thus \eqref{eq:dF1}  satisfies \eqref{eq:dF110}. On the other hand, if we assume $2\dF|\cos(\theta)|\geq D$, the solution of \eqref{eq:dF110} is obtained as follows:
\begin{align}
    \label{eq:dF2}
    \dF=8D^2\sin^2(\theta)/\lambda
\end{align}
Besides, it can be verified from \eqref{eq:dF110} that by increasing $\theta$, the Fraunhofer distance undergoes a continuous switching from \eqref{eq:dF1} to \eqref{eq:dF2}  at  $\theta=\pi/2 - \theta^{\mathrm{F}}$ and then back to \eqref{eq:dF1} at $\theta=\pi/2 + \theta^{\mathrm{F}}$ wherein the following equality holds:
\begin{align}
    \frac{8}{\lambda}D^2\sin^2(\theta)=\frac{D}{2|\cos(\theta)|}, \ \mathrm{for}\ \theta=\pi/2 - \theta^{\mathrm{F}}
\end{align}
which verifies that $\theta^{\mathrm{F}}=\pi/2-F^{-1}(\frac{\lambda}{2D})$. Considering this, together with \eqref{eq:dF1} and \eqref{eq:dF2}, the Fraunhofer distance $\dF$ is obtained as \eqref{eq:dF234}.
\end{proof}
\begin{definition}
    As seen in Theorem \ref{th:dF}, for $|\pi/2 - \theta|\geq \theta^{\mathrm{F}}$, we have $d^F=8D^2\sin^2 \theta/\lambda$. We call $\theta^{\mathrm{F}}$ the {\it Fraunhofer angle}.
\end{definition}
\begin{property}
   The Fraunhofer angle is tightly approximated 
    \begin{align}
        \label{eq:tFApprox}
        \theta^{\mathrm{F}}\approx \frac{1}{2}\sin^{-1}\left({\frac{\lambda}{8{D}}} \right)
    \end{align}
   
\end{property}
 The validity of the approximation can be easily verified by considering $\theta^{\mathrm{F}}=\pi/2 - 
 F^{-1}(\frac{\lambda}{2D})\approx \pi/2-F_0^{-1}(\frac{\lambda}{2D})$ for $D/\lambda \geq 0.5$ where
    \begin{multline}
        F(\theta)=8|\cos(\theta)|\sin^2(\theta) 
        = 
        8\cos(\theta)\sin(\theta)
        = {4}\sin(2\theta)
        \\
        \equiv F_0(\theta),
        \forall \theta \in [\underbrace{{\pi}/{2}- F^{-1}(1)}_{82.7^\circ},{\pi}/{2})
    \end{multline}
    Fig .\ref{fig:thetaF} depicts the exact value of $\theta^{\mathrm{F}}=\pi/2-F^{-1}(\lambda/2D)$ as well as the approximated value from \eqref{eq:tFApprox} with various numbers of array elements $N$ corresponding to different values of $\frac{D}{\lambda}$ for a ULA antenna with half-wavelength inter-element spacing.
   \begin{figure}
    \centering
    \includegraphics[width=244pt]{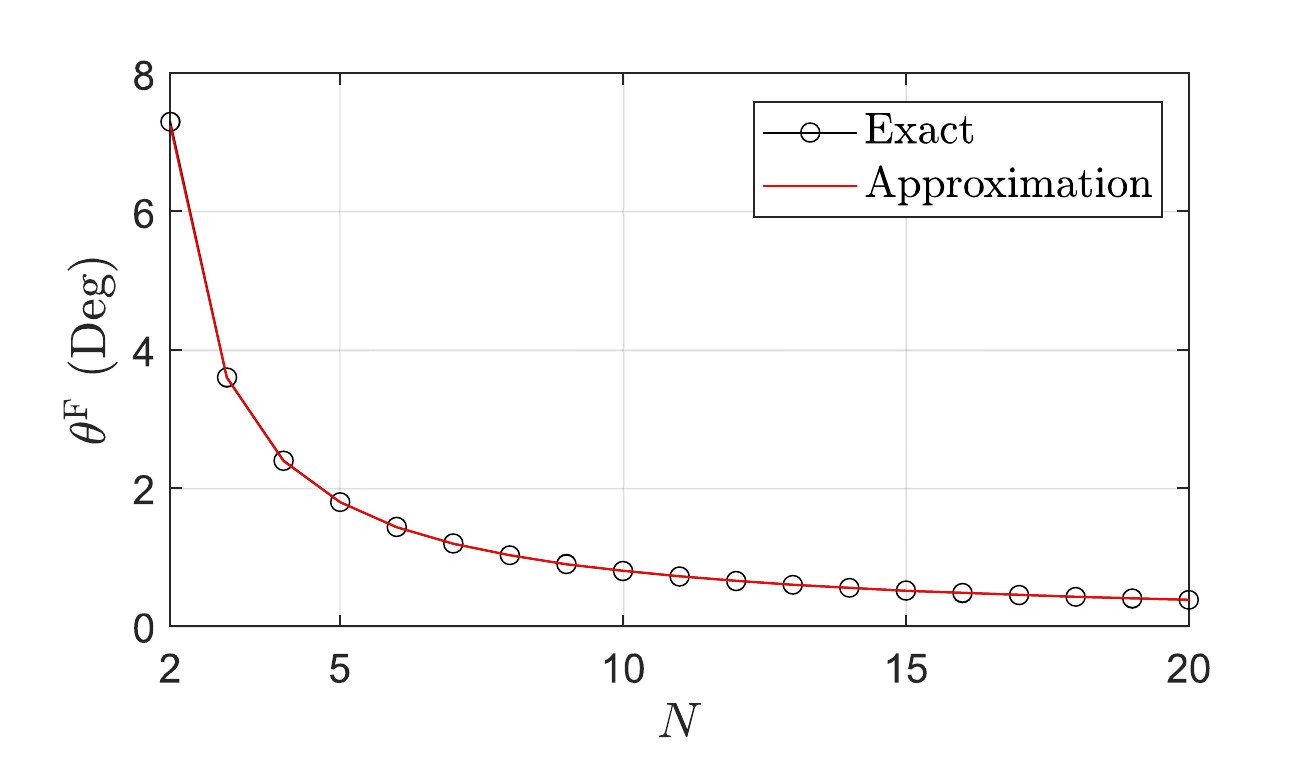}
    \caption{The exact and approximated value of $\theta^{\mathrm{F}}$ for various number of array elements $N$ for a ULA antenna with half-wavelength inter-element spacing.}
    \label{fig:thetaF}
\end{figure}

\begin{corollary}
   The maximum Fraunhofer distance is 
   \begin{align}
\label{eq:dFbar}
    \dFMax
    &=
    8D^2\cos^2(\theta^{\mathrm{F}})/\lambda\approx 4 \times \dFzeroMax
\end{align}
where $\dFzeroMax$ is the maximum boresight Fraunhofer distance of a single-element antenna having the same diameter as the intended phased array.
\end{corollary}
\begin{proof}
From \eqref{eq:dF234} we have $\dF(\theta)=\dF(\pi-\theta)$, and thus we only consider $\theta\in[0,\pi/2]$. Note that the term $\dF=8D^2\sin^2(\theta)/\lambda$ in \eqref{eq:dF234} is an increasing function of $\theta$ for $\theta\in(0,\pi/2-\theta^{\mathrm{F}}]$, and thus, we only need to show that $\dF<8D^2\sin^2(\theta)/\lambda\big|_{\theta=\pi/2-\theta^{\mathrm{F}}}= 8D^2\cos^2(\theta^{\mathrm{F}})/\lambda$  for $\theta\in[ \pi/2- \theta^{\mathrm{F}},\pi/2)$. It can be easily verified that \eqref{eq:dF1} is a monotonically decreasing function of $\theta$ for $F^{-1}(1)\leq F^{-1}(\frac{\lambda}{2D}) \leq \theta \leq \pi/2$ and thus the maximum Fraunhofer distance is obtained as \eqref{eq:dFbar}.
\end{proof}

Fig. \ref{fig:regions} schematically depicts how the Fraunhofer distance is initially increased and then decreased when moving from the on-boresight to the off-boresight scenario. 
The Fraunhofer distance versus $\theta$ for a single-element antenna (i.e., $N=1$) as well as ULAs with $N\in \{3,10,40\}$ is depicted in Fig. \ref{fig:Fraunhofer_array} by using \eqref{eq:dF234} as well as running full-wave simulations using HFSS software. It is seen as $N$ increases (corresponding to higher dimension of the antenna), the value of $\theta^{\mathrm{F}}$ rapidly decreases toward zero, and for all values of $\theta$ other than the very small region around the principal axis in the angular domain corresponding to $\theta\in[0,\pi/2-\theta^{\mathrm{F}}] \cup [\pi/2+\theta^{\mathrm{F}},\pi] $, the Fraunhofer distance is $8D^2\sin^2(\theta)/\lambda$ which is 4 times the Fraunhofer distance function $\dFzero$  corresponding to a single-element antenna with the same diameter $D$ having a centered reference point on the feed location. The full-wave simulations obtained from an array of half-wavelength dipoles verify the characterized formulation of the Fraunhofer distance. As shown in the figure, when the mutual coupling $\mu$ is set to zero, the results obtained from (11) and the full-wave simulation are in perfect agreement. By increasing $\mu$, the Fraunhofer distance is seen to be increased. The reason behind that can be expressed as follows. The mutual coupling between two adjacent elements in half-wavelength inter-element spacing dipole arrays is capacitive \cite{balanis2016antenna}; therefore, an {\bf increase} in the phase difference occurs between two neighboring elements if a non-zero mutual coupling is incorporated. This phase shift is observed for any pair of consecutive elements in the array, resulting in a corresponding increase in the maximum phase delay.
Therefore, if the reference point is located at the Fraunhofer distance characterized by considering $\mu=0$, the incorporation of a mutual coupling $\mu>0$ increases the maximum differential phase delay from $\pi/8$ to a higher value. This can be compensated by increasing the distance of the reference transmitting point source to return the corresponding delay at the receiving antenna aperture to $\pi/8$. 

\begin{figure}
    \centering
    \includegraphics[width=244pt]{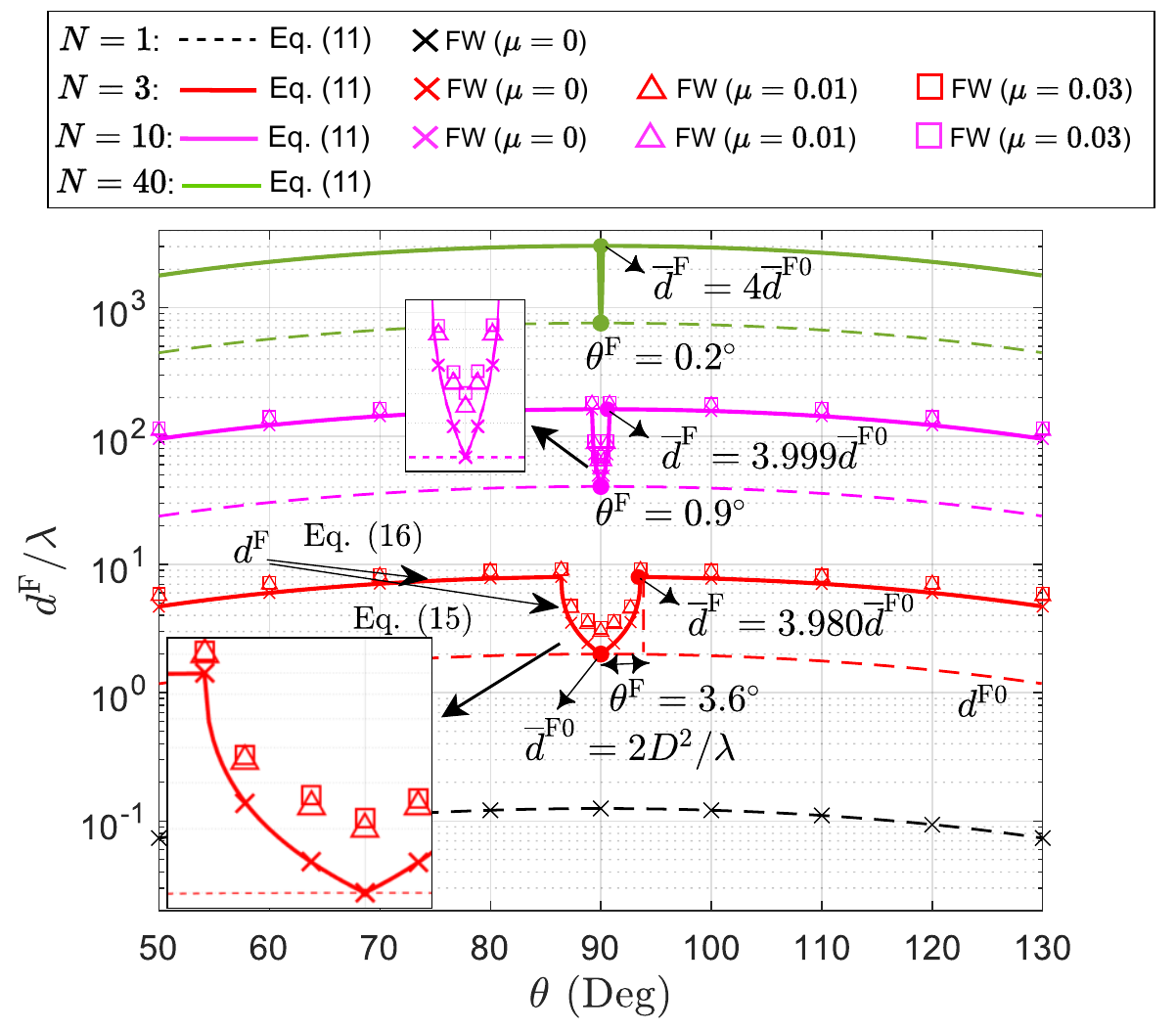}
    \caption{Fraunhofer distance per wavelength ($d^{\mathrm{F}}/\lambda$)  versus observation angle $\theta$ for $N$-element ULA with half-wavelength inter-element spacing (i.e., $D=\frac{(N-1) \lambda}{2}$). The dashed lines relate the principal axis Fraunhofer distance function $d^{\mathrm{F0}}$. The value of $\dFzeroMax$ for all curves corresponds to $2D^2/\lambda$ at $\theta=90$. The results from full-wave (FW) simulations are provided as well by considering an array of half-wavelength dipole elements with different values of mutual coupling ($\mu$). 
    }
    \label{fig:Fraunhofer_array}
\end{figure}

Considering that access points (APs) are typically installed at higher elevations compared to user equipment (UEs), leading to off-boresight scenarios, Fig. \ref{fig:Fraunhofer_h}  illustrates the impact of the relative height of the serving phased array antenna (denoted by $h$) on the maximum near-field coverage distance (denoted by $\overline{d}$). Three setups are considered including 2 different carrier frequencies and 2 different aperture sizes. For instance, in the case of a carrier frequency of $f=28$ GHz and an aperture dimension of $0.7\times 0.7 \ \mathrm{m^2}$, increasing $h$ from zero (on-boresight scenario) to $\overline{h}=35$ cm results in an increase in the maximum near-field coverage distance $\overline{d}$ from 183 m to 731 m. Further increasing $h$ from 35 cm to about 100 m has a negligible effect on $\overline{d}$. Finally, increasing $h$ from 100 m to 290 m reduces $\overline{d}$ to zero due to the dominance of far-field propagation.

\begin{figure}
    \centering
    \includegraphics[width=244pt]{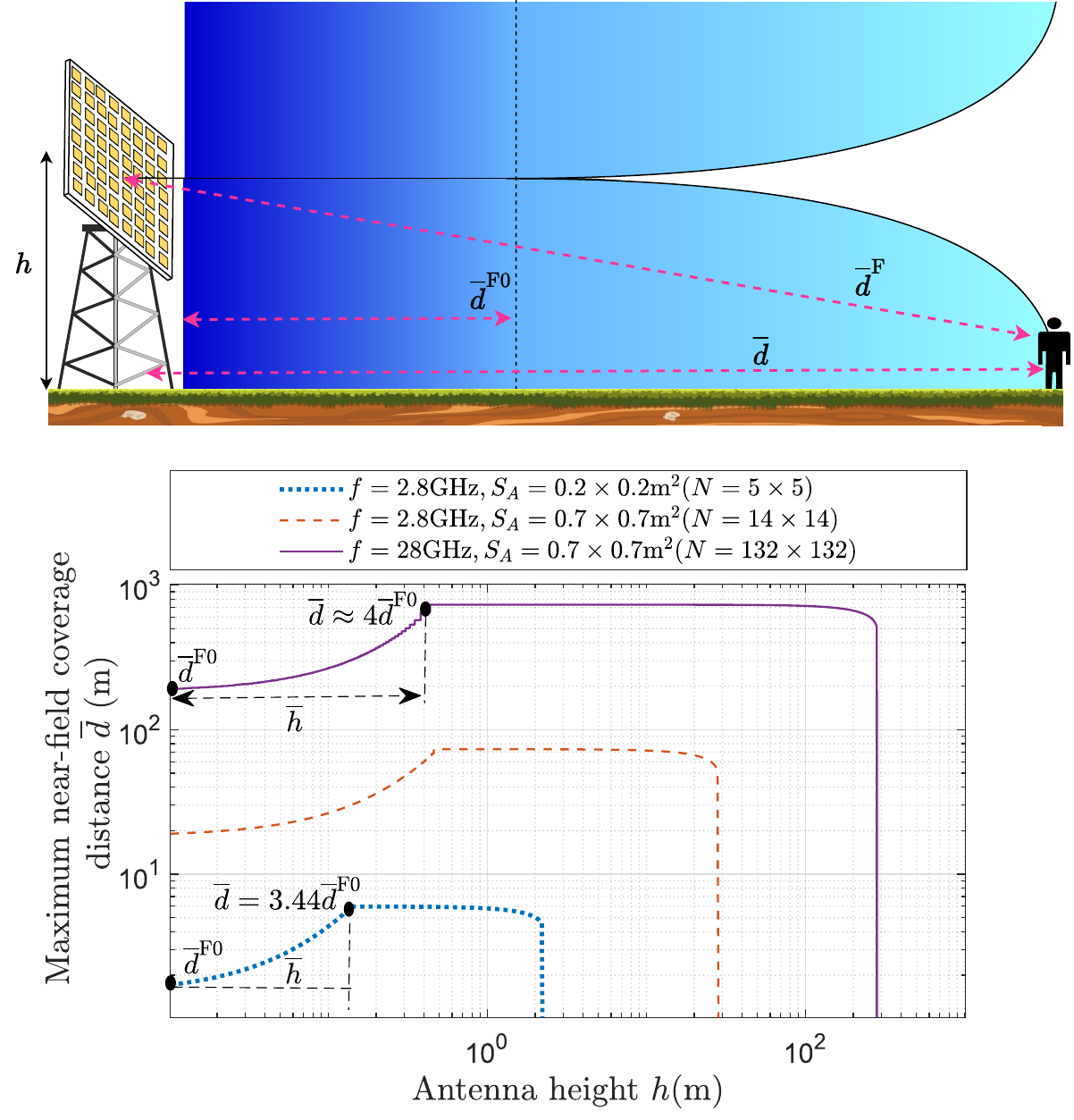}
    \caption{Maximum near-field coverage distance $\overline{d}$ for a phased array antenna at height $h$, considering the operating frequency $f\in\{2.8,28\}$ GHz and aperture area $S_A\in \{0.2\times 0.2, 0.7 \times 0.7\}\mathrm{m}^2$}
    \label{fig:Fraunhofer_h}
\end{figure}

\section{Characterization of Radial Beamfocusing Region}
\label{sec:radial}
In this section, we explore the near-field beamfocusing for phased array antennas in the radial domain. First, we present the formal definitions relating to characterizing the radial focal region and beamfocusing in the radial domain. Then we show that the maximum ratio transmission (MRT) beamformer based on the near-field uniform spherical wave model results in a focal gap between the desired focal point and the achieved focal point. We discuss the challenge of finding a closed-form solution for a radial focal point beamformer and propose an efficient algorithm for achieving this in a few iterations.

\subsection{Characterization and properties of beamfocusing in the radial domain}
Consider a transmit antenna array such as the UPA or ULA antennas illustrated in Fig. \ref{fig:system_model}. First, we start investigating the radial beamfocusing in the far-field region. In this region, we have $g_n(\rbold)\approx  \frac{1}{r},\forall n$, and besides, the polarization loss and effective aperture loss are only a function of $\theta$ and $\phi$. Thus we have
\begin{subequations}
\begin{align}
    G_1(\rbold,\sbold_n)=G_1(\rbold,\sbold_{\widetilde{n}})\equiv G_1(\theta,\phi)
    \\
    G_2(\rbold,\sbold_n)=G_2(\rbold,\sbold_{\widetilde{n}})\equiv G_2(\theta,\phi)
\end{align}
\end{subequations}
Therefore, from \eqref{eq:E126}, the received signal is formulated from the GNC model as
\begin{align}
\label{eq:E129}
    y(\rbold,\bbold)= 
    \underbrace{\sqrt{\frac{ G_1(\theta,\phi)G_2(\theta,\phi)}{4\pi r^2}} 
    e^{-j\frac{2\pi}{\lambda}r}
    }_{y_0(\rbold)}
    \underbrace{\sum_{n=1}^{N}e^{j \left(-\frac{2\pi}{\lambda}   (r_n-r)+\beta_n\right) }}_{AF}
\end{align}
where AF is the array factor. It can be verified that for the far-field, it suffices to consider only the first-order term of the Taylor series of $r_n$ in \eqref{eq:rn_versus_r}, which leads to the following:
\begin{align}
\label{eq:rn_minus_r}
    r_n-r\approx x_n^{\mathrm{s}}\sin(\theta)\cos(\phi)+z_n^{\mathrm{s}}\cos(\theta)
\end{align}
From \eqref{eq:rn_minus_r} and \eqref{eq:E129} 
the magnitude of the received signal can be written as
\begin{align}
    \label{eq:absoluty_far}
    \textrm{far-field:\ \ \ }|y(\rbold,\bbold)|=A(\theta,\phi)\times \frac{1}{r}|AF(\theta,\phi,\bbold)|
\end{align}
where $A(\theta,\phi)=\sqrt{G_1(\theta,\phi)G_2(\theta,\phi)/4\pi}$ and $AF(\theta,\phi,\bbold)$ is the {\it array factor}.
It is seen in \eqref{eq:absoluty_far} that the array factor in the far-field is only a function of $\theta$ and $\phi$, and it does not depend on $r$. For the near-field, \eqref{eq:E129} can be written as
\begin{multline}
\label{eq:yarraynear}
      y(\rbold,\bbold)= 
      \underbrace{\sqrt{\frac{ G_1(\theta,\phi)G_2(\theta,\phi)}{4\pi r^2}} 
    e^{-j\frac{2\pi}{\lambda}r}
    }_{y_0(\rbold)}
    \times
    \\
    \underbrace{\sum_{n=1}^{N}
    \sqrt{
    \frac{ G_1(\rbold,\sbold_n)G_2(\rbold,\sbold_n) }{G_1(\theta,\phi)G_2(\theta,\phi)}\times\frac{r^2}{r_n^2}
    }
    e^{j \left( -\frac{2\pi}{\lambda}  (r_n-r)+\beta_n\right) }}_{AF}
\end{multline}
Therefore, the array factor is a function of $r$, as well as $\theta$ and $\phi$, and thus we have
\begin{align}
    \label{eq:absoluty_near}
    \textrm{near-field:\ \ \ }|y(\rbold,\bbold)|
    &=
    A(\theta,\phi)\times \frac{1}{r}|AF(r,\theta,\phi,\bbold)|
    \notag
    \\
    &\equiv
     A(\theta,\phi)\times \frac{1}{r}|AF(\rbold,\bbold)|
\end{align}
It is seen from \eqref{eq:absoluty_far} that for the far-field, $|y(\rbold,\bbold)|\propto \frac{1}{r}$ is always a decreasing function of $r$ and thus there exists no directivity in the radial domain. For the near-field, however, this does not hold; as seen in \eqref{eq:absoluty_near}, the relation of the array factor to $r$ can result in a non-monotonic behavior of $|y(\rbold,\bbold)|$.



\begin{definition}
Given the angular values $\theta$ and $\phi$, we call the point $\rbold^{\mathrm{f}}=(r^{\mathrm{f}},\theta,\phi)$ a {\it radial focal point} (RFP) corresponding to some focal beamfocusing vector $\bbold^{\mathrm{f}}$, if it corresponds to a local maximum of the absolute value of the radiated signal in the radial domain. i.e., 
\begin{align}  
\label{eq:7795}
    \frac{\partial |\frac{1}{r}AF(\rbold,\bbold^{\mathrm{f}})|}{\partial r}\bigg|_{r=r^{\mathrm{f}}}\hspace{-18pt}=0 \ \ \ \mathrm{and} \ \ \  \frac{\partial^2 |\frac{1}{r}AF(\rbold,\bbold^{\mathrm{f}})|}{\partial r^2}\bigg|_{r=r^{\mathrm{f}}}\hspace{-18pt}<0 
\end{align}
The RFP $\rbold^{\mathrm{f}}=(r^{\mathrm{f}},\theta,\phi)$ is called a {\it 3dB radial focal point} (3dB RFP), if there exists some $\underline{r}^{\mathrm{f}}<r^{\mathrm{f}}$ and $\overline{r}^{\mathrm{f}}>r^{\mathrm{f}}$
for which we have
\begin{align}
\label{eq:7796}
    \left| \frac{1}{r} AF(\rbold,\bbold^{\mathrm{f}}) \right|
    \geq
    \frac{1}{\sqrt{2}}
     \left| \frac{1}{r^{\mathrm{f}}} AF(\rbold^{\mathrm{f}},\bbold^{\mathrm{f}}) \right|, \forall r\in [\underline{r}^{\mathrm{f}}, \overline{r}^{\mathrm{f}}]
\end{align}
The minimum possible interval corresponding to $\overline{r}^{\mathrm{f}}-\underline{r}^{\mathrm{f}}$ is called the {\it 3dB radial focal depth}.
\end{definition}
\begin{definition}
\label{def:dRF}
For a given $\theta$ and $\phi$, let $\mathcal{D}^{\mathrm{f}}(r)$ and $\mathcal{D}^{\mathrm{f,3dB}}(r)$ be the  the domain of the beamfocusing vectors for which \eqref{eq:7795} and  \eqref{eq:7796} hold respectively. The {\it radial focal distance} denoted by $d^{\mathrm{RF}}$ and {\it 3dB radial focal distance} denoted by $d^{\mathrm{RF,3dB}}$ are maximum values of $r$ corresponding to a feasible RFP and 3dB RFP respectively. i.e.,
\begin{align}
    &d^{\mathrm{RF}}
    =
    \max\{r\}, \ \mathrm{s.t.}\ \mathcal{D}^{\mathrm{f}}(r)\neq \O
    \\
    &d^{\mathrm{RF,3dB}}
    =
    \max\{r\}, \ \mathrm{s.t.}\ \mathcal{D}^{\mathrm{f,3dB}}(r)\neq \O
\end{align}
\end{definition}
From Definition \ref{def:dRF}, for radial domain distance values $r\leq d^{\mathrm{RF}}$, we can find at least one beamfocusing vector   $\bbold_0\in \mathcal{D}^{\mathrm{f}}(r)$ for which the point $(r,\theta,\phi)$ is an RFP; and for $r>d^{\mathrm{RF}}$, there exists no beamfocusing vector to realize an RFP. A similar discussion holds for $d^{\mathrm{RF,3dB}}$. It should be noted that $d^{\mathrm{RF,3dB}}<d^{\mathrm{RF}}$.


    

\begin{property}
    For any array antenna, and any beamfocusing vector $\bbold$, we have $d^{\mathrm{RF}}<d^{\mathrm{F}}$.
\end{property}   

This property is directly verified by observing  \eqref{eq:absoluty_far}, which implies that in the far-field, there exists no beamfocusing vector $\bbold$ corresponding to the solution of $\frac{\partial |y(\rbold,\bbold)|}{\partial r}=0$. For example, for phased array antennas having small-size antenna patch elements and using Fresnel approximation, the authors of \cite{bjornson2021primer} have proved that the 3dB depth of focus is possible at distances lower than ${\dF}/{10}$.  

    Given the angles $\theta$ and $\phi$, let ${\bbold}^{\mathrm{MR}}$ be the {\it maximum ratio transmission} (MRT) beamformer corresponding to maximum signal amplitude, and ${\bbold}^{\mathrm{f}}$ be the {\it maximum radial focusing} beamformer corresponding to the focal point with highest amplitude at distance $r$. From \eqref{eq:absoluty_near} and Definition \eqref{def:dRF},  ${\bbold}^{\mathrm{MR}}$ and  ${\bbold}^{\mathrm{f}}$ are obtained as follows:
\begin{align}
    \label{eq:betabar}
     {\bbold}^{\mathrm{MR}}
     &=\argmax_{\bbold}{|AF(\rbold,\bbold)|} 
     \\
     \label{eq:betafbar}
   {\bbold}^{\mathrm{f}}
   &=\argmax_{\bbold}{|AF(\rbold,\bbold)|}, \ \bbold\in \mathcal{D}^{\mathrm{f}}(r)
\end{align}
    The solution to \eqref{eq:betabar} is trivial and is obtained as $b^{\mathrm{MR}}_n=e^{j\frac{2\pi}{\lambda}r_n}, \forall n$, however, it is not straightforward to obtain the solution to \eqref{eq:betafbar} since finding the domain $\mathcal{D}^{\mathrm{f}}(r)$ corresponding to the feasibility of \eqref{eq:7795} is a challenging issue and necessitates the incorporation of GNC or NUSW models for solving \eqref{eq:7795} which is generally intractable.
 The important question that arises is whether the beamformer $\bboldbarMR$ corresponding to MRT transmission for point $\rboldbar=[\rbar,\theta,\phi]^{\mathrm{T}}$  results in a dominant radial focal point at some location $\rbold=[r,\theta,\phi]^{\mathrm{T}}$. 
The answer to this question provides a practical method to obtain the beamformer corresponding to a desired radial focal point at  $\rboldf$ as will be elaborated in Section \ref{sec:radial_alg}. 
\begin{property}
\label{prop:MRRR}
\label{prop:mfp}
    Let $\bboldbarMR$ be the MRT beamformer corresponding to point $\rboldbar=(\rbar,\theta,\phi)$, and let $\rboldf=(\rf,\theta,\phi)$ be the  radial focal point closest to $\overline{\rbold}$ in the radial domain, called as {\it dominant radial focal point} with respect to $\overline{\rbold}$. We have $|y(\rbold^{\mathrm{f}},\bboldbarMR)|>|y(\overline{\rbold},\bboldbarMR)|$ and $\rf<\rbar$.
\end{property}


To justify Property \ref{prop:MRRR}, note that $\bboldbarMR$ results in the highest signal amplitude at point $\rboldbar=(\rbar,\theta,\phi)$. Increasing the radial distance from $\overline{r}$  to some distance $r>\overline{r}$ decreases the signal level because the completely coherent arrival of the signals from different array elements does not exist anymore, and besides, the channel gains of all antenna elements are decreased due to higher distance. 
Therefore, the existence of a dominant focal point is not possible at $r> \overline{r}$. 
Decreasing the radial distance from $\overline{r}$ to some value $r<\overline{r}$, however, has two opposite effects. On one hand, changing the location results in the violation of a fully coherent arrival of the signals, leading to a weaker signal, and on the other hand, increases the array factor due to lower values of $r_n$ in \eqref{eq:yarraynear}. 
The point where the former attenuating factor begins to dominate the latter intensifying one is the focal point $\rbold^{\mathrm{f}}$, wherein the signal magnitude is higher than $\overline{\rbold}$ due to the monotonically decreasing property of $|y|$ around $\rbar$ in the radial domain.   
\begin{definition}
    Let $\bboldbarMR$ be the MRT beamformer corresponding to the point $\rboldbar$, and $\rboldf$ be the resulting dominant radial focal point. The {\it radial focal gap} corresponding to the point $\rboldbar$ is defined as $G(\rboldbar)=\rbar-\rf$.
\end{definition}
\subsection{Radial Focal Gap and Spot Beamfocusing in the asymptotic scenario}
In this section, we study the radial focal gap in the asymptotic case where the number of antenna elements tends to infinity. The following Property will be used in providing the proof for the next Theorem.
\begin{property} 
\label{prop:1}
  Considering the USW normalized channels, for two points $\rbold_1\neq \rbold_2$, the array response has 2D asymptotic orthogonality for the ULA and 3D asymptotic orthogonality for the UPA if the number of antennas $N$ is sufficiently large. This property is mathematically expressed as follows:
\begin{align}
    \lim_{N\rightarrow \infty} \frac{1}{N} |\abold^{\mathrm{T}} (\rbold_1) \abold^*(\rbold_2)|=0
\end{align}
where $\rbold_i=(r_i,\theta_i)$ in the 2D domain, and $\rbold_i=(r_i,\theta_i,\phi_i)$ in the 3D domain.
\end{property}

\begin{proof}
    Please refer to \cite{10123941}.
\end{proof}
It can be demonstrated through various simulations that the radial focal gap $\rbar-\rf$ approaches arbitrarily small values for many antenna structures if they have a sufficiently large number of array elements. The following theorem provides a sufficient condition for this.
\begin{theorem}
\label{th:gap}
     Let $\rboldf=(\rf,\theta,\phi)$ be the dominant radial focal point resulting from the MRT beamformer $\bboldbarMR$ corresponding to the point $\rboldbar$, and assume that the antenna elements gain vector $\gbold(\rbold)$ can be approximated by a set of piece-wise uniform gain vectors (i.e., $\gbold(\rbold)_{N\times 1}\approx[\gbold_1^{\mathrm{T}}(\rbold),\gbold_2^{\mathrm{T}}(\rbold), ....,\gbold_M^{\mathrm{T}}(\rbold)]^{\mathrm{T}}$, where  $\gbold_m(\rbold)$ is $\frac{N}{M}\times 1$, and   $[\gbold_m(\rbold)]_{n_1}=[\gbold_m(\rbold)]_{n_2}\equiv g_m,\forall n_1,n_2,m$. In the asymptotic case, the radial focal gap vanishes, i.e., $         \lim_{N\rightarrow\infty} \left(\overline{r}- r^{\mathrm{f}} \right)=0$,
     and besides, a 3D spot-like beamfocusing is achieved around $\rboldbar$.
\end{theorem}
\begin{proof}
    The proof can be followed by employing Property \ref{prop:1}. The details are omitted here due to conciseness.
\end{proof}

It is worth noting that Theorem \ref{th:gap} establishes a sufficient condition for ensuring the disappearance of the asymptotic radial focal gap. Even though the conditions of the theorem may not hold for many antenna structures such as ULA or UPA, we can prove that the asymptotic disappearance of radial focal gap under the NUSW model holds for many antenna array structures. In what follows we present two theorems to verify this for ULAs and UPAs. 

\begin{theorem}
\label{th:ULA_asymptotic_ratio}
    Consider a ULA antenna and two  points  $\rbold=[r,\theta]^{\mathrm{T}}$ and $\rboldbar=[\rbar,\thetabar]^{\mathrm{T}}$ in the 2D domain in the near-field region of the antenna 
    and let $\bboldbarMR$ be  the MRT beamformer corresponding to $\rboldbar$. Under the NUSW model, the normalized signal amplitude at any point $\rbold\neq \rboldbar$ ($r\neq \rbar$ or $\theta\neq \thetabar$) asymptotically tends to zero as $N$ increases. i.e.,
    \begin{align}
        \lim_{N\rightarrow \infty}
        \kappa_N=\lim_{N\rightarrow \infty}\frac{\big|y(\rbold,\bboldbarMR)\big|}{\big|y(\rboldbar,\bboldbarMR)\big|}=0
    \end{align}
\end{theorem}
\begin{proof}
     See \textbf{Appendix A.}
\end{proof}

\begin{theorem}
    \label{th:UPA_asymptotic_ratio}
    Consider a UPA antenna and two points  $\rbold=[r,\theta,\phi]^{\mathrm{T}}$ and $\rboldbar=[\rbar,\thetabar,\overline{\phi}]^{\mathrm{T}}$ in the 3D domain in the near-field region of the antenna. 
    Let $\bboldbarMR$ be the MRT beamformer corresponding to $\rboldbar$. Under the NUSW model, the normalized signal amplitude at any point $\rbold\neq \rboldbar$ ($r\neq \rbar$ or $\theta\neq \thetabar$ or  $\phi\neq \overline{\phi}$) asymptotically tends to zero as $N$ increases. i.e.
    \begin{align}
        \lim_{N_1,N_2\rightarrow \infty}
        \kappa_{N_1,N_2}=\lim_{N_1,N_2\rightarrow \infty} \frac
{\big|y(\rbold,\bboldbarMR)\big|}{\big|y(\rboldbar,\bboldbarMR)\big|}=0
    \end{align}
\end{theorem}
\begin{proof}
  See \textbf{Appendix B.}
\end{proof}

The following corollary can be directly deducted from Theorems \ref{th:ULA_asymptotic_ratio} and \ref{th:UPA_asymptotic_ratio}.
\begin{corollary}
    Under the NUSW near-field channel model of UPA and ULA antennas, in the asymptotic case where $N$ tends to infinity, the following properties hold:
    \begin{enumerate}[(a)]
        \item The radial focal gap tends to zero.
        \item The focal region reduces in size and concentrates into:
        \begin{itemize}
            \item a small spot-like 3D focal point applied to the UPA.
            \item a thin 2D focal ring applied to the ULA.
        \end{itemize}
    \end{enumerate}
\end{corollary}
\begin{figure}
    \centering
    \includegraphics[width=254pt]{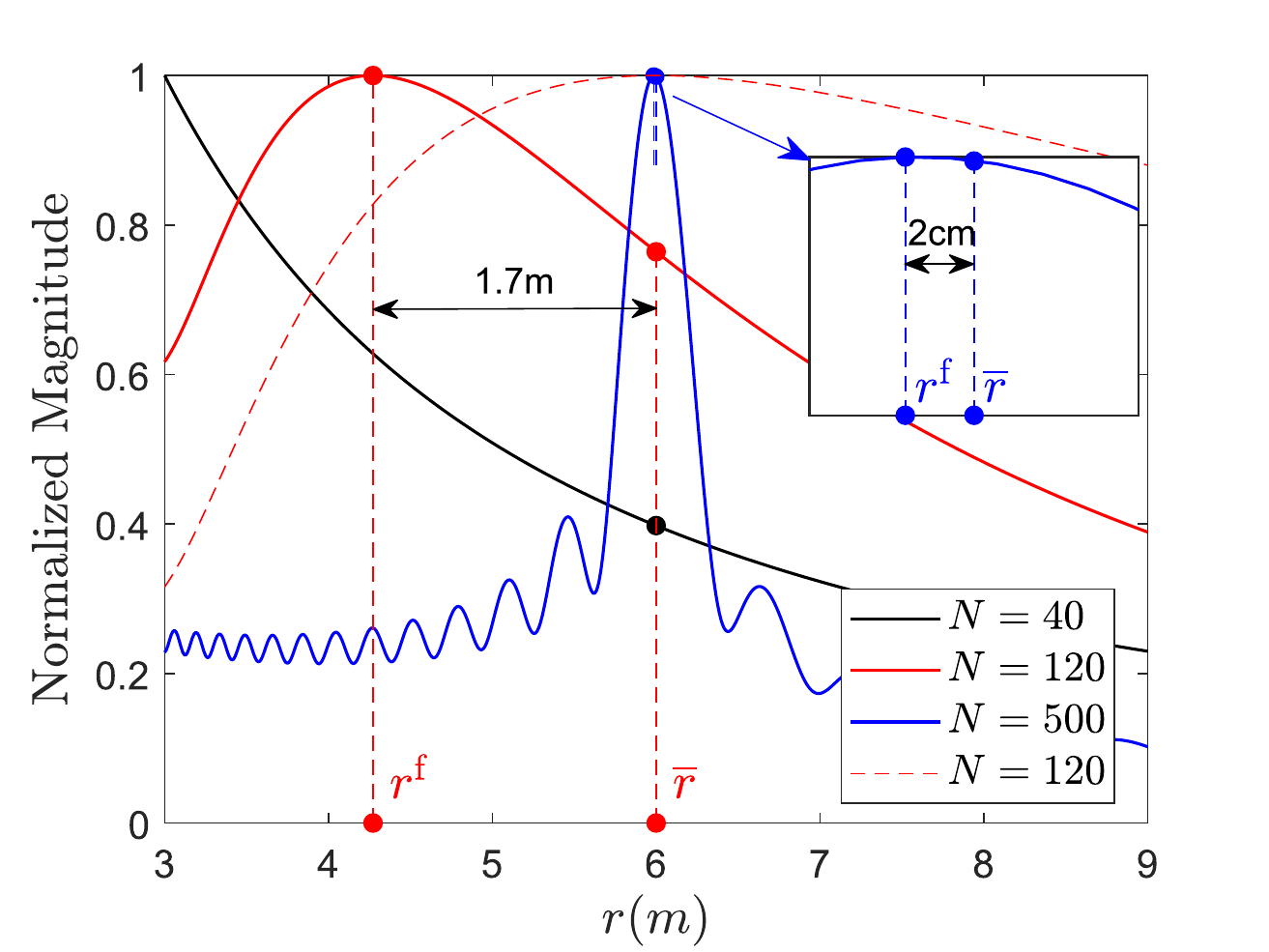}
    \caption{Radial beamfocusing for $N$-element ULA antenna. The frequency is 28GHz, and the inter-element spacing is half wavelength. The solid curves and the dashed curve correspond to employing the NUSW and USW channel models.}
    \label{fig:BF}
\end{figure}

3D near-field spot beamfocusing (SBF) which can be realized by ULA and UPA antennas in the asymptotic case, has many potential applications not only in wireless communication and high-power and safe wireless power transfer (WPT) \cite{monemi2023towards} but also in health and medical sensing (e.g., stimulating specific neurons through neuromodulation \cite{krames2018neuromodulation}), 
semiconductor and THz technology (e.g., high speed turning on/off nano switch arrays through casting spot-like power on each switch element \cite{ou2020tunable}), 
etc.

Fig. \ref{fig:BF} depicts the normalized signal magnitude of a ULA consisting of $N$ antenna elements in the radial domain, where $N\in\{40,120,500\}$. We have considered the beamformer vector $\bbold=\bboldbarMR$ corresponding to a point with a distance of $\rbar=6$m from the antenna center located on the boresight. First, it is seen that for $N=40$ (corresponding to $\dF=8.6$m), the MRT beamformer results in no focal point at any point $\rf<\rbar$ in the radial domain. 
Considering $N=120$, a dominant focal point is achieved at $\rf=4.3$m, exhibiting a radial focal gap of 1.7m. To illustrate that the efficacy of the USW is limited in analyzing the radial domain focusing and capturing the achieved radial focal gap, the normalized magnitude for $N=120$ obtained using the USW model is also depicted by the red dashed curve. It is observed that this model yields inaccurate results, and even more concerning, it fails to capture the radial gap between $\rf$ and $\rbar$, and consequently, employing the GNC or NUSW becomes necessary.  
Finally, following Theorem~\ref{th:ULA_asymptotic_ratio}, it is observed that for a significantly large value of $N$ (equal to 500), the gap between $\rboldbar$ and the achieved radial focal point, $\rboldf$, is substantially reduced, reaching a mere 2cm.

\begin{figure}
    \centering
    \includegraphics[width=254pt]{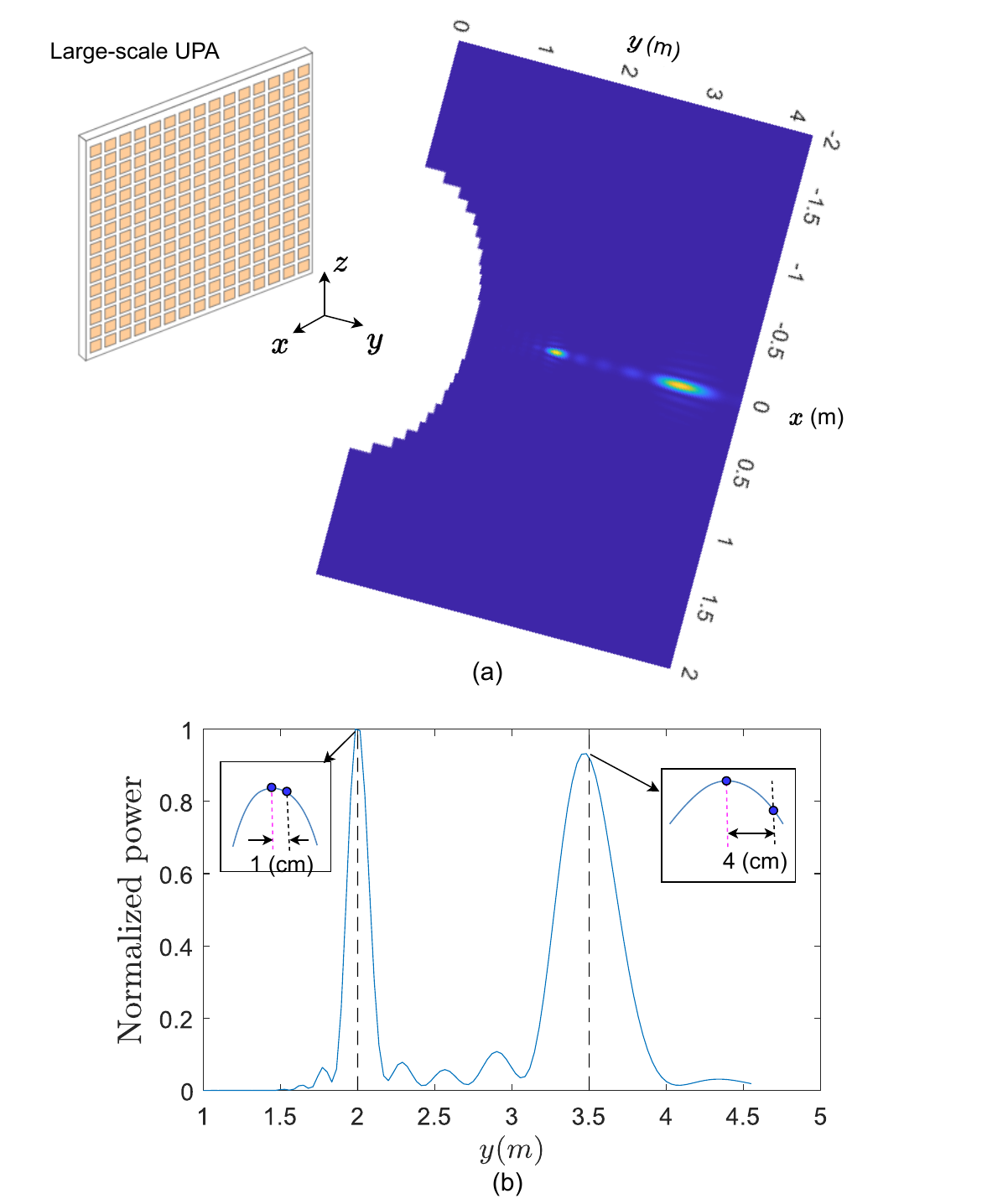}
    \caption{Multi-focal radial and angular beamfocusing considering the NUSW model for a large-scale $100\times 100$ UPA with a fully connected MRT beamformer having two RF chains. The frequency is 28GHz, the inter-element spacing is half wavelength, and array elements are considered to be isotropic. The antenna is located on the $xz$-plane centered is (0,0,0.5) and the DFPs are located at the $(0,2,0)$ and $(0,3.5,0)$.}
    \label{fig:BF_UPA}
\end{figure}

The radial focal gap between the DFP and AFP for a multi-focal 3D beamfocusing scenario using a large-scale planner array is investigated in Fig. \ref{fig:BF_UPA}. We have considered a $100\times 100$ UPA, having a fully connected MRT beamformer with two RF chains, where the DFPs are located at points $(0,2,0)$ and $(0,3.5,0)$. It is seen that 3D multi-focal beamfocusing is achievable through large-scale planner arrays both in radial and angular domains. Besides, it is observed that the 3dB depth of focus as well as the radial focal gap between the AFP and DFP are greater for the DFPs with higher distances. This is expected as the farther DFPs are from the aperture, the closer to the Fraunhofer boundary, and thus less near-field spherical curvature is experienced. Here, for the near and far DFPs, we have a gap of $1$cm and $4$cm respectively.

\subsection{Beamformer design to realize radial focal point in the non-asymptotic scenario}
\label{sec:radial_alg}
Theorems \ref{th:gap}, \ref{th:ULA_asymptotic_ratio} and \ref{th:UPA_asymptotic_ratio} reveal that under the NUSW near-field channel model, for most antenna structures, if we have an array with a very large number of elements, the radial focal point is simply realized through analog MRT beamformer corresponding to the point $\rboldbar \approx \rboldf$. This leads to a minimal gap between the DFP and AFP.
\begin{algorithm}[t]
		\caption{\small\!: Calculation of beamfocusing vector corresponding to the focal point $\rboldf$}
		\begin{algorithmic}[1]
           


            \Statex
            \hspace{-14pt}
              {\bf Output:}
            Beamfocusing vector $\bbold^*$ corresponding to a desired radial focal point $\rboldf=(\rf,\theta,\phi)$;
            
            \Statex
            \hspace{-14pt}
            {\bf Initialization:}
            \State
            Let $\alpha\leftarrow 1+\eps$, where $\eps \ll 1$ is a small number, $k\leftarrow 1$; {\color{black}$\hat{y}'_1 \leftarrow -1$};  $\rbar_{1}\leftarrow \rf$; 
 {\color{black}$\hat{y}_{1}\leftarrow y(\rboldf,\bboldbar_{1})$}, where $\bboldbar_{1}$ is the maximal beamforming vector corresponding to the point $\rbold_{1}=\rboldf=(\rbar_{1},\theta,\phi)$ ;
            \Statex
            \hspace{-14pt}
            {\bf Main procedure:}

            \Repeat

            \State Let $k\leftarrow k+1$;
            
            \State Let  $\rbar_{k} \leftarrow \alpha \rbar_{k-1}$,
            $\hat{y}_{k}\leftarrow \left|y(\rboldf,\bboldbar_{k})\right|$, where $\bboldbar_{k}$ is the maximal beamforming vector corresponding to the point $\rboldbar_{k}=(\rbar_{k},\theta,\phi)$;

            \State Let $\hat{y}'_k\leftarrow (\hat{y}_k-\hat{y}_{k-1})/(\rbar_k-\rbar_{k-1})$;

            \Until{$\hat{y}'_k \times \hat{y}'_{k-1} \leq 0$}
            
            \State Obtain the optimal $\rbar^*\in[\rbar_k, \rbar_{k-2}]$ corresponding to the  zero-point of $\hat{y}'$  at $r=\rf$ through bisection search algorithm, with initial inputs $\rbar_k$ and $\rbar_{k-2}$;

            \State Obtain $\bbold^*$ as the maximal beamfocusing vector corresponding to the point $\rboldbar^*=(\rbar^*,\theta,\phi)$;

    \end{algorithmic}
\end{algorithm}
\begin{figure}
    \centering
    \begin{tabular}{c} 
     \includegraphics[width=254pt]{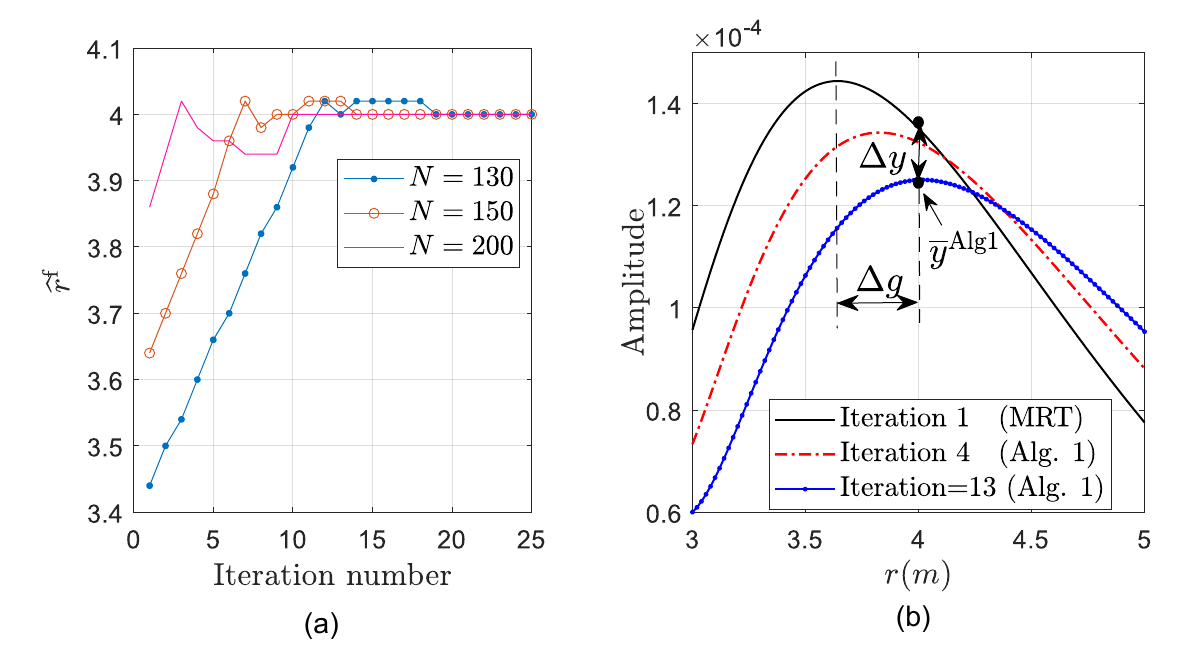}
    \end{tabular}
    \caption{The convergence of Algorithm 1 for implementing the radial beamfocusing at the focal distance $r^{\mathrm{f}}=4$m on the boresight of the antenna for an $N$-element ULA. The left figure illustrates the achieved focal distance per iteration for $N\in\{130,150,200\}$, and the right figure plots the signal amplitude in the radial domain at three iterations for $N=150$.
    } 

    \label{fig:BF_Converge}
\end{figure}
A practical problem to be addressed is how to find a beamfocusing vector that results in a focal point at some DFP with the distance $\rf$ in the radial domain when the number of array elements is not too large. In general, there exists no analytical closed-form solution to this problem based on the GNC or UNSW models, even for the simple ULA case. 
Considering the monotonically increasing property of the signal magnitude when decreasing the radial distance from $\rbar$ to $\rf$, as described in Property \ref{prop:mfp}, we propose Algorithm 1 which practically obtains the beamfocusing vector corresponding to some desired focal point at a distance $\rf$ through a limited number of iterations. As seen in the algorithm,  
since we have  $\rf<\rbar$ for $\bbold=\bboldbarMR$, instead of setting the MRT beamfocusing vector $\bboldbar$ corresponding to $\rf$, we set it for some points with higher distances in an iterative manner according to steps 4 and 5 in the loop, and measure numerically the slope of the received signal magnitude at the observation point from two successive measurements. As soon as a change in the sign of the slope is detected for some $\rbar_k$, the desired point for the optimal MRT beamformer can be found as some $\rbar^*\in[\rbar_k, \rbar_{k-2}]$. At this step, the bisection search algorithm finds the exact value of $\rbar^*$.
It is seen that in the initialization phase, we have set the initial slope of $|y_1|$ denoted by $\hat{y}'_1$ to -1; this is because the initial dominant radial focal point is realized for MRT beamformer $\bboldbar_1$ at Step 1 at some point $\Tilde{\rbold}_1$ with $\Tilde{r}_1<\rbar_1$ and with an amplitude $|y(\Tilde{\rbold}_1,\bboldbar_1)|>|y({\overline{\rbold}}_1,\bboldbar_1)|$. Considering this, and noting that the algorithm is only influenced by the sign of $\hat{y}'_k$, we can set an arbitrary negative value to $\hat{y}'_1$ in the initialization phase. 
To show the convergence of the proposed algorithm in realizing a radial focal point at the DFP, we have considered ULAs with $N\in\{130,150,200\}$ and a focal point on the boresight of the antenna at $r^f=4m$. Fig. \ref{fig:BF_Converge}-a depicts the convergence of the achieved focal point $\widehat{r}^{\mathrm{f}}$ to the desired value. It is seen that for a larger value of $N$, corresponding to smaller gaps between $\rf$ and $\rbar$, the algorithm converges with a higher speed, and besides, the higher the number of array elements, the smaller, the initial distance between $\widehat{r}^f$ and DFP. In addition, for $N=150$, we have depicted the radial pattern for three iterations in Fig. \ref{fig:BF_Converge}-b to demonstrate how the patterns in the radial domain are updated. 
Finally, noting that Algorithm 1 and the MRT follow different objectives, the performance comparison is presented in Fig. \ref{fig:BF_Compare} in terms of two different metrics.  Consider a circle for representing the focal region, centered at the DFP and having radius $R$ (where we have considered here $R=0.5$ m), and let $\mathcal{F}$ denote the portion of the circle area wherein the signal amplitude is lower than that corresponding to the DFP. While Algorithm 1 is seen to achieve the target peak power at the DFP resulting in $\mathcal{F}^{\mathrm{Alg1}}=1$, a decay in the amplitude level is experienced on the other hand as observed in both Fig. \ref{fig:BF_Converge}-b and the orange-colored lines of Fig. \ref{fig:BF_Compare}. It is seen how increasing the number of antenna elements lowers the amplitude gap and highers the value of $\mathcal{F}$, leading to equal performance measures of Algorithm 1 and MRT in the asymptotic scenario.

\begin{figure}
    \centering
    \begin{tabular}{c} 
     \includegraphics[width=244pt]{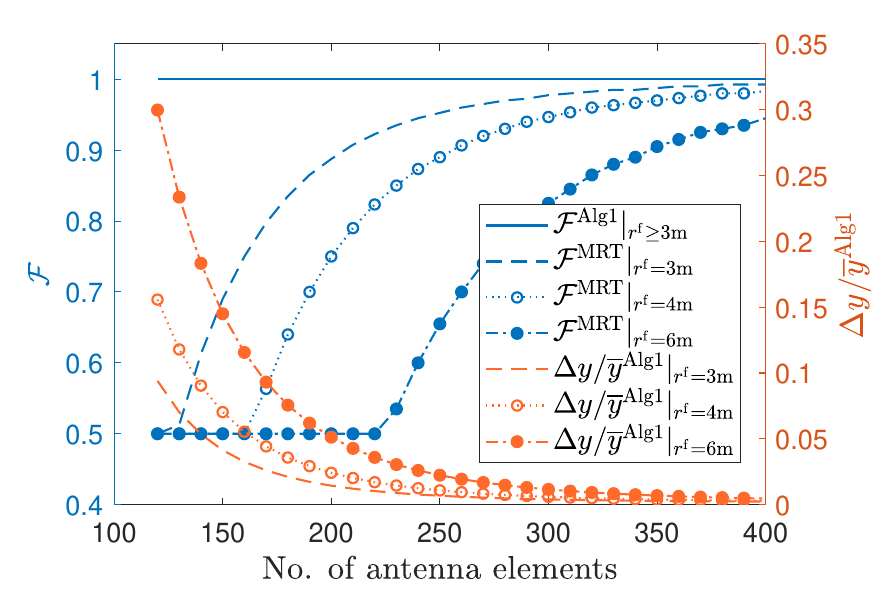}
    \end{tabular}
    \caption{Performance comparison of Algorithm 1 and MRT for different values of radial focal distance $r^\mathrm{f}$ and No. of antenna elements $N$. 
    }
    \label{fig:BF_Compare}
\end{figure}

To analyze the computational complexity, note that Algorithm 1 contains two iterative procedures, corresponding to the loop expressed in Steps 2 to 6, and the bisection search stated in Step 7. These two procedures have limited iterations and generally converge fast, as seen in Fig. \ref{fig:BF_Converge}. 
    The complexity for each iteration is determined by the dominant computationally complex term corresponding to the MRT beamformer $\bboldbar_k$, which is $\mathcal{O}(N)$ when considering a fully digital beamformer with complete channel state information (CSI). The computational complexity increases when employing hybrid beamforming structures, due to the overhead associated with analog and digital beamformers.
    
\section{Characterization of Non-Radiating Region}
\label{sec:nonradiating}
In this part, we investigate the non-radiating distance of phased array antennas.  The non-radiating distance $\dNR$ is a distance very close to the antenna where the near-field propagation models should be fully considered. Hence we require the very exact expressions of electric and magnetic fields to calculate $d^{\mathrm{NR}}$. Consider a ULA of diameter $D$ consisting of $N$ center-fed thin dipole antennas each having a diameter $D^{\mathrm{s}}$ located on the $z$ access and centered at the origin as seen in Fig. \ref{fig:ULA_dR}. We aim to calculate the active and reactive powers at some point on the boresight with distance $r$ to the antenna. To do so, as seen in the figure, first consider an infinitesimal dipole located on point $(0,0,z')$ at the Cartesian coordinate system with electric current density $\Ibold^{(n)}(z')=I^{(n)}(z')\aboldhat_z$. The electric potential function relating to an infinitesimal current distribution on $z'$ is obtained as follows \cite{balanis2016antenna}:
\begin{align}
    d\Abold^{(n)}(r,\phi,z)=\aboldhat_z \frac{\mu I_n(z') d{z'}}{4 \pi R'} e^{-jkR'}
\end{align}
where $\mu$ is the permeability and $R'=\sqrt{r^2+(z')^2}$ is the distance of the observation point to the infinitesimal antenna at $(0,0,z')$. The electric and magnetic fields resulting from $d{\Abold}^{(n)}$ are obtained respectively from the following Maxwell's equations:
\begin{figure}
    \centering
    \includegraphics[width=254pt]{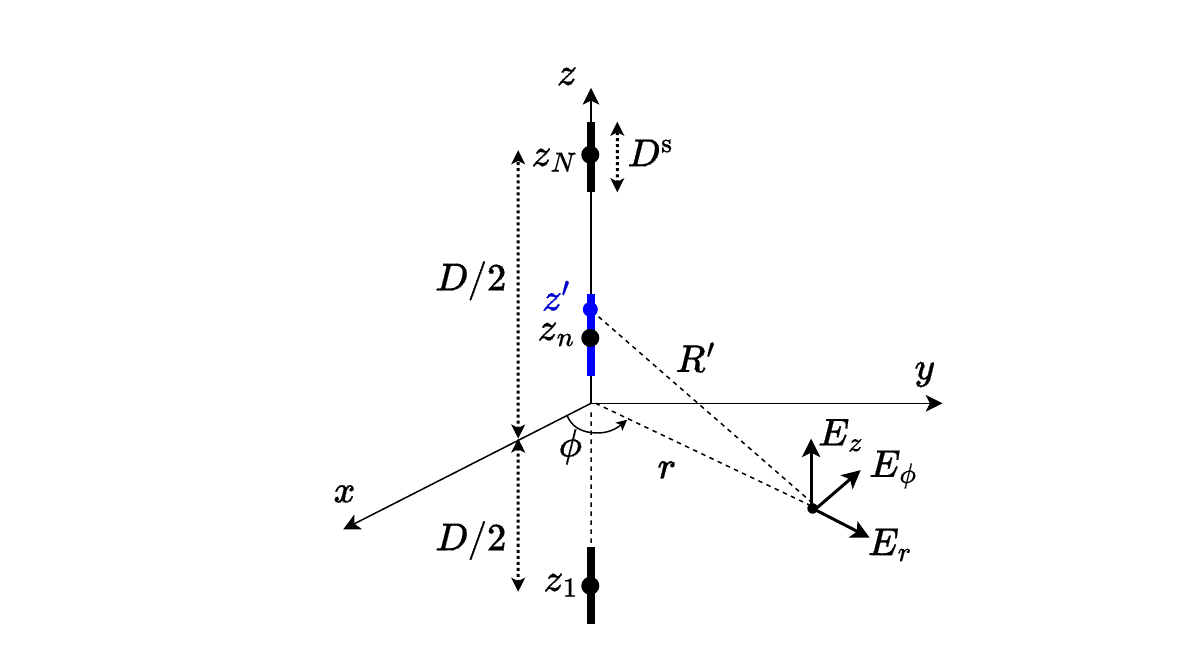}
    \caption{Dipole ULA structure for calculating $d^{\mathrm{NR}}$.}
    \label{fig:ULA_dR}
\end{figure}
\begin{subequations}
\label{eq:Maxwells}
        \begin{align}
        \label{eq:Maxwells_H}
        d\Hbold^{(n)}
        &=
        \frac{1}{\mu} \nabla \times d{\Abold^{(n)}}
        \\
        \label{eq:Maxwells_E}
        \nabla \times d\Hbold^{(n)}
        &=
        j\omega  \eps \times d\Ebold^{(n)}
    \end{align}
\end{subequations}
By taking the curls in the Cylindrical coordinate system and after doing some mathematical manipulations, $d\Hbold^{(n)}$ and $d\Ebold^{(n)}$ are derived as follows.
\begin{subequations}
\begin{align}
    d\Hbold^{(n)}
    &=
    dH_{\phi}^{(n)}\aboldhat_{\phi}
    \\
 d\Ebold^{(n)}
 &=
 dE_r^{(n)}\aboldhat_{r}
 + 
 dE_z^{(n)}\aboldhat_{z}
\end{align}
\end{subequations}
where
\begin{subequations}
\begin{align}
    dH_{\phi}^{(n)}
    &=
    \frac{r I^{(n)}(z')d{z'}}{4\pi}e^{-jkR'}
    \left( 
    jkR'^{-2}+R'^{-3}
    \right)
    \\
      \notag
    dE_r^{(n)}
    &=
    \frac{jrk\eta I^{(n)}(z')d{z'}R'^{-3}}{4\pi} e^{-jkR'}(z-z') \times
    \\
    &
    \hspace{60pt}
    \left(
    1-3jk^{-1}R'^{-1}-3k^{-2}R'^{-2}
    \right)
     \\
     \notag
    dE_z^{(n)}
    &=
    \frac{\eta I^{(n)}(z')d{z'}R'^{-2}}{4\pi} e^{-jkR'} \times
    \\
    &
    \hspace{-10pt}
    \left(
    2-j\left(2k^{-1}+kr^2\right)R'^{-1}-3r^2R'^{-2}+j3r^2k^{-1}R'^{-3}
    \right)
\end{align}
\end{subequations}
The current distribution of the dipole antenna element $n$ can be well approximated as follows\cite{balanis2016antenna}:
\begin{multline}
    I^{(n)}(z')=I_0^{(n)}
    \sin
    \left[
    k\left(
    \frac{D^{\mathrm{s}}}{2}-|z'-z_n|
    \right)
    \right]
    \\
    , z_n-\frac{D^{\mathrm{s}}}{2} \leq z' \leq z_n+\frac{D^{\mathrm{s}}}{2}
\end{multline}
To express the electric field $\Ebold^{(n)}$ and magnetic field $\Hbold^{(n)}$ of the $n$'th dipole antenna in a simple form, first, we define the functions $F$ and $G$ as follows.
\begin{subequations}
\begin{align}
    &\hspace{-5pt}F_\alpha^{(n)} \equiv F_{\alpha}(D^{\mathrm{s}},z_n,r)=
    \notag
    \\
    & \hspace{-15pt}
    \hspace{10pt}\int_{z_n-\frac{D^{\mathrm{s}}}{2}}^{z_n+\frac{D^{\mathrm{s}}}{2}}
\frac
{\sin
    \left[
    k\left(
    \frac{D^{\mathrm{s}}}{2}-|z'-z_n|
    \right)
    \right]}
{\left(r^2 +(z')^2\right)^{\alpha/2}}
e^{-jk(r^2 +(z')^2)^{0.5}}dz'
\\
    & \hspace{-5pt}G_\alpha^{(n)} \equiv G_{\alpha}(D^{\mathrm{s}},z_n,r)=
    \notag
    \\
    & \hspace{-10pt}\int_{z_n-\frac{D^{\mathrm{s}}}{2}}^{z_n+\frac{D^{\mathrm{s}}}{2}}
\frac
{\sin
    \left[
    k\left(
    \frac{D^{\mathrm{s}}}{2}-|z'-z_n|
    \right)
    \right]}
{\left(r^2 +(z')^2\right)^{\alpha/2}}
e^{-jk(r^2 +(z')^2)^{0.5}}z'dz'
\end{align}
\end{subequations}
One can verify that the following equality holds:
\begin{subequations}
\begin{align}
    F_{\alpha}^{(n)}(D^{\mathrm{s}},z_n,r)
    &=
    \lambda^{1-\alpha}\widetilde{F}^{(n)}_{\alpha}(\widetilde{D}^\mathrm{s},\widetilde{z}_n,\widetilde{r})
    \\
    G_{\alpha}^{(n)}(D^{\mathrm{s}},z_n,r)
    &=
    \lambda^{2-\alpha}\widetilde{G}^{(n)}_{\alpha}(\widetilde{D}^\mathrm{s},\widetilde{z}_n,\widetilde{r})
\end{align}
\end{subequations}
where $\widetilde{D}^{\mathrm{s}}=D^{\mathrm{s}}/\lambda$, $\widetilde{z}_n=z_n/\lambda$,  $\widetilde{r}=r/\lambda$, and
\begin{subequations}
\begin{align}
    &\widetilde{F}_\alpha^{(n)} 
    \equiv \widetilde{F}^{(n)}_{\alpha}(\widetilde{D}^\mathrm{s},\widetilde{z}_n,\widetilde{r})=
    \notag
    \\
    &\int_{-\frac{\widetilde{D}^\mathrm{s}}{2}}^{\frac{\widetilde{D}^\mathrm{s}}{2}}
\frac
{\sin
    \left[
    2\pi\left(
    \frac{\widetilde{D}^\mathrm{s}}{2}-|z'|
    \right)
    \right]}
{\left(\widetilde{r}^2 +(z'+\widetilde{z}_n)^2\right)^{\alpha/2}}
e^{-j2\pi(\widetilde{r}^2 +(z'+\widetilde{z}_n)^2)^{0.5}}dz'
\\
 &\widetilde{G}_\alpha^{(n)} 
    \equiv \widetilde{G}^{(n)}_{\alpha}(\widetilde{D}^\mathrm{s},\widetilde{z}_n,\widetilde{r})=
    \notag
    \\
    &\!\!\!\!\int_{-\frac{\widetilde{D}^\mathrm{s}}{2}}^{\frac{\widetilde{D}^\mathrm{s}}{2}}
\frac
{\sin
    \left[
    2\pi\left(
    \frac{\widetilde{D}^\mathrm{s}}{2}-|z'|
    \right)
    \right]}
{\left(\widetilde{r}^2 +(z'+\widetilde{z}_n)^2\right)^{\alpha/2}}
e^{-j2\pi(\widetilde{r}^2 +(z'+\widetilde{z}_n)^2)^{0.5}}z'dz'
\end{align}
\end{subequations}
One can verify that $H^{(n)}_{\phi}$, $E^{(n)}_r$, and $E^{(n)}_z$ can be obtained in terms of functions $\widetilde{F}$ and $\widetilde{G}$ as follows:
\begin{align}
    H^{(n)}_{\phi}
    =
    \int_{z_n-\frac{D^{\mathrm{s}}}{2}}^{z_n+\frac{D^{\mathrm{s}}}{2}}
    dH^{(n)}_{\phi} =\frac{jkI_0^{(n)} \widetilde{r} }{4\pi} 
    \left[
    \widetilde{F}_2^{(n)} - j(2\pi)^{-1}\widetilde{F}_3^{(n)}
    \right]
\end{align}
\begin{multline}
    E^{(n)}_{r}
    =
    \int_{z_n-\frac{D^{\mathrm{s}}}{2}}^{z_n+\frac{D^{\mathrm{s}}}{2}}
    dE^{(n)}_{r} 
    =
    \frac{j\eta k I_0^{(n)} \widetilde{r} }{4\pi} 
    \bigg[
    -\widetilde{G}_3^{(n)}+
    \\
     3j(2\pi)^{-1}\widetilde{G}_4^{(n)} +3(2\pi)^{-2}\widetilde{G}_5^{(n)}
    \bigg]
\end{multline}
\begin{multline}
    E^{(n)}_{z}
    =
    \int_{z_n-\frac{D^{\mathrm{s}}}{2}}^{z_n+\frac{D^{\mathrm{s}}}{2}}
    dE^{(n)}_{z} 
    = \frac{\eta k I_0^{(n)} }{8\pi^2} 
    \bigg[
    2\widetilde{F}_2^{(n)} - 
    \\
    \hspace{-6pt}
    2j\left((2\pi)^{-1}+ 
    \pi \widetilde{r}^2\right)\widetilde{F}_3^{(n)} -3\widetilde{r}^2\widetilde{F}_4^{(n)} +3j(2\pi)^{-1}\widetilde{r}^2\widetilde{F}_5^{(n)}
    \bigg]
\end{multline}
Ignoring the mutual coupling between array elements, the total power is then obtained as follows:
\begin{multline}
\label{eq:P3299}
    \boldsymbol{P}=\frac{1}{2}\boldsymbol{E}\times \boldsymbol{H}^*
    =
    \\
    \frac{1}{2}
    \left[
    \sum_{n=1}^N \left( E_r^{(n)}\aboldhat_r + E_z^{(n)}\aboldhat_z \right)
    \right]
    \times
    \left[
    \sum_{n=1}^N  H_{\phi}^{*(n)}\aboldhat_\phi
    \right]
\end{multline}
Finally, the active and reactive powers are obtained respectively as
\begin{subequations}
\begin{align}
    \boldsymbol{P}^{\mathrm{a}}
    &=
    Re\{\boldsymbol{P}\}
    = Re \left\{P_z \right\} \aboldhat_z
    +
    Re\left\{P_r\right\}\aboldhat_r
    \\
     \boldsymbol{P}^{\mathrm{r}}
    &=
    Im\{\boldsymbol{P}\}
    = Im \left\{P_z \right\} \aboldhat_z
    +
    Im\left\{P_r\right\}\aboldhat_r
\end{align}
\end{subequations}
where $
    P_z =
    \frac{1}{2} 
    \sum_{n=1}^N E_r^{(n)} \sum_{n=1}^N H_{\phi}^{*(n)}
    $, and $
    P_r =
     \frac{1}{2} 
    \sum_{n=1}^N E_z^{(n)} \sum_{n=1}^N H_{\phi}^{*(n)}$. 
\begin{figure}
    \centering
    \includegraphics[width=254pt]{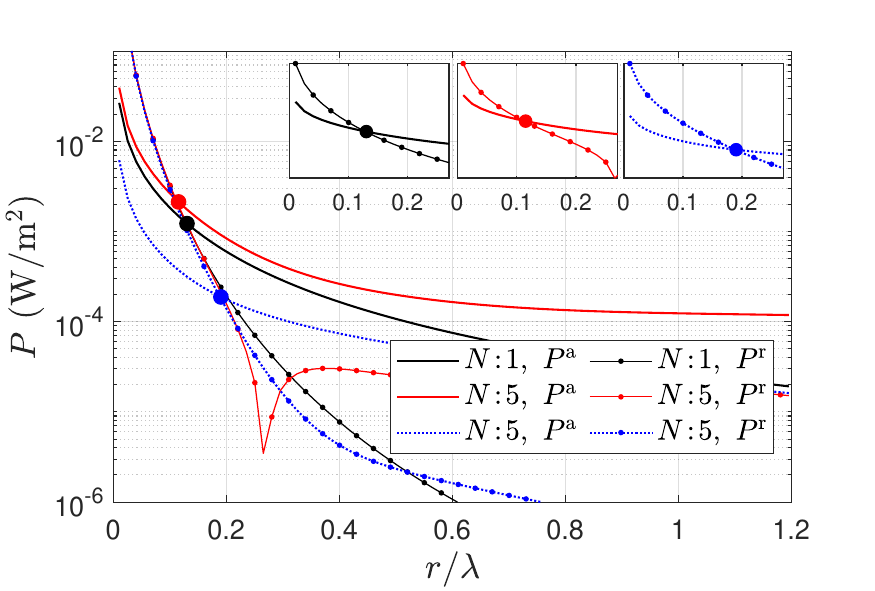}
    \caption{The magnitude of active power density $P^\mathrm{a}$ and reactive power density $P^{\mathrm{r}}$  considering a single-element dipole antenna, as well as 5-element phased arrays with current densities corresponding to $I_0^{(n)}=I_0, \forall n$ (solid lines), and $I_0^{(n)}=(-1)^nI_0, \forall n$ (dotted lines). The inter-element spacing is half-wavelength and $D^{\mathrm{s}}=0.25\lambda$}
    \label{fig:dR_r}
\end{figure}
The non-radiating distance $d^{\mathrm{NR}}$ is obtained as the maximum distance $r$ at which the magnitude of the active and reactive powers are equal. i.e.,
\begin{multline}
    d^{\mathrm{NR}}=\max\{r\}, \mathrm{s.t.} 
    \\
      Re\{P_z\}^2 + Re\{P_r\}^2 
    =
     Im\{P_z\}^2 + Im\{P_r\}^2 
\end{multline}

\begin{figure}
    \centering
    \includegraphics[width=254pt]{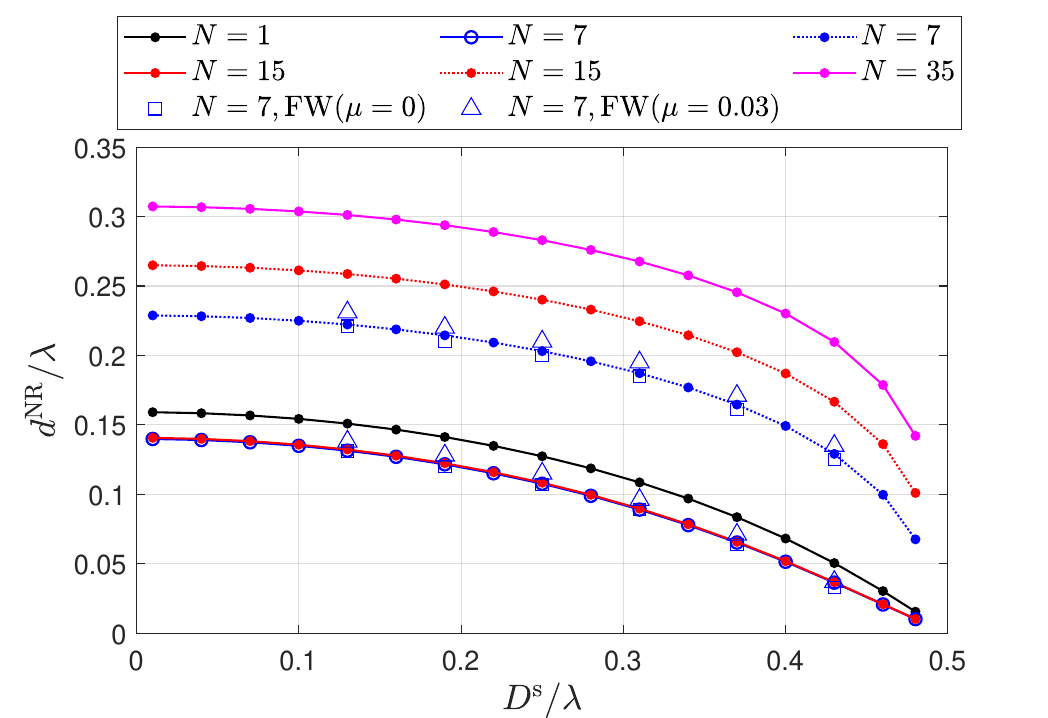}
    \caption{The non-radiating distance per $\lambda$ corresponding to in-phase transmissions with $I_0^{(n)}=I_0, \forall n$ (solid lines) and non-in-phase transmissions with $I_0^{(n)}=(-1)^n I_0, \forall n$ (dotted lines) versus the diameter of antenna element (per $\lambda$) for various number of antenna elements. The square and triangle marks are obtained by full-wave simulations for $N=7$ for both in-phase and non-in-phase scenarios considering the mutual coupling of $\mu=0$ and $\mu=0.03$ respectively.}
    \label{fig:dR_Ns}
\end{figure}
To evaluate $\dNR$,
two scenarios are considered for simulations. In the first one, we consider that the transmitted signals from the array elements are all of the same phase, corresponding to $I_0^{(n)}=I_0, \forall n$, and in the second one we consider the phase difference between the signal of each antenna element and the neighboring element is $180^{\circ}$ corresponding to $I_0^{(n)}=(-1)^n I_0, \forall n$. We call the first and second cases in-phase and non-in-phase scenarios respectively.  The simulation results have been obtained using our derived formulations which have been verified by full-wave simulations through HFSS software as well for $N=7$ considering $\mu=0$ and $\mu=0.03$.
Fig. \ref{fig:dR_r} illustrates how the magnitude of the active and reactive power density decreases as the distance $r$ is increased considering a single-element antenna, as well as 5-element in-phase and 5-element non-in-phase arrays, considering half-wavelength interelement spacing.  It is seen that in all cases the reactive power density dominates the active power at distances very close to the antenna, and after a very short distance (much less than a wavelength), the active power dominates for both in-phase and non-in-phase scenarios.  
To investigate how $D^{\mathrm{s}}$ impacts the radiating distance, we have considered in Fig. \ref{fig:dR_Ns} various values of $D^{\mathrm{s}}$ ranging from 0.01 to 0.48 and obtained the corresponding non-radiating distance $\widetilde{d}^{\mathrm{NR}}$ per $\lambda$ for different numbers of array elements $N$ and both scenarios of in-phase and non-in-phase transmissions. Firstly, it is seen that for all cases, as $D^{\mathrm{s}}$ increases toward half a wavelength, the non-radiating distance decreases toward zero for all values of $N$ and both in-phase and non-in-phase cases. This is because of the well-known property that a half-wavelength is highly efficient and the power is mainly radiated actively.  Besides, it is seen that for a single dipole antenna ($N=1$) and a very small value of $D^{\mathrm{s}}$, $d^{\mathrm{NR}}$ tends to $0.159\lambda$; this is equivalent to a distance of $\lambda/2\pi$ at which the active and reactive power density are equal for infinitesimal dipole, as mathematically proven in \cite{balanis2016antenna}. 
On the other hand, it is seen that the non-radiating distance for various array configurations lies consistently below half a wavelength, which is significantly smaller than the Fresnel distance. This observation invalidates the notion that the non-radiating distance is constrained by the Fresnel zone, a common misconception prevalent in many existing works.  For example, if we consider an antenna array with $N=20$ element and half-wavelength interelement spacing, the maximum Fresnel distance is $1.75\sqrt{\frac{(10\lambda)^3}{\lambda}}\approx 55 \lambda$ which is much higher than half a wavelength.

\section{Conclusions}
\label{sec:conclusions} 
In this paper, we characterized the off-boresight Fraunhofer distance as well as the non-radiating distance for phased array antennas. We showed how the Fraunhofer distance extends to about 4 times when moving from on-boresight to off-boresight angles. 
Furthermore, we elaborated on the formal definition and characterization of the radial focal region based on the general near-field channel model. We discussed how the MRT beamformer designed based on the USW near-field channel model leads to a radial focal gap between the desired and achieved focal points. We presented an algorithm to realize the desired radial focal point in the non-asymptotic scenario (when the number of array elements is not very large) at the exact DFP location by removing the impact of the radial focal gap. In total, our analysis emphasizes the importance of proper formulations and a deeper understanding of various near-field propagation regions.
\begin{appendices}
    \section{Proof of Theorem \ref{th:ULA_asymptotic_ratio}}
Consider a ULA with interelement spacing $\delta$ located on the $z$ axis as shown in Fig. \ref{fig:system_model}-b, having a large number of elements $N$, and excited with MRT beamformer $\bboldbarMR$ corresponding to point $\rboldbar=(\rbar,\thetabar)$. Assuming that the central antenna element is located at the origin corresponding to $n=0$, the distance from element $n$ to point $\rbold$ is $r_n=\sqrt{r^2-2\delta\cos(\theta)rn +\delta^2n^2}$. By using the Taylor series $\sqrt{1+x}\approx 1+ \frac{x}{2}-\frac{x^2}{8}$, we may approximate the phase component $\frac{2\pi}{\lambda}(r_n-\rbar_n)$ as
 \begin{align}
 \label{eq:5676wvsdgf}
        & \frac{2\pi}{\lambda}
        \big( 
        \sqrt{r^2-2\delta\!\cos(\theta)rn + \delta^2 n^2 } 
        - \sqrt{\rbar^2-2\delta\!\cos(\thetabar)\rbar n + \delta^2 n^2 }
        \big)
         \notag
         \\
         & \hspace{100pt}\approx{2\pi}(r-\rbar)/\lambda +\eta_1 n + \eta_2 n^2
    \end{align}
   where $\eta_1=\frac{2\pi}{\lambda}\delta \left(\cos(\overline{\theta}) - \cos(\theta) \right)$ and $\eta_2=\frac{\pi\delta^2}{\lambda}\left( \frac{\sin^2(\theta)}{r}-
     \frac{\sin^2(\overline{\theta})}{\rbar}
     \right)$. 
     Considering that the term $2\pi(r-\rbar)/\lambda$ in \eqref{eq:5676wvsdgf} is independent of $n$, the signal amplitude at point $\rbold=(r,\theta)$ can be obtained as
    \begin{align}
        \left|y(\rbold,\bboldbarMR)\right|
        &=\left|\sum_{-(N-1)/2}^{(N-1)/2} \frac{ e^{-j\frac{2\pi}{\lambda}(r_n-\rbar_n)}}{\sqrt{r^2-2\delta\cos(\theta) rn + \delta^2 n^2 }} \right|
        \notag
        \\
        &\approx
        \left|\int_{-\frac{N}{2}}^{\frac{N}{2}} \frac{ e^{-j(\eta_1 n+\eta_2 n^2)}}{\sqrt{r^2-2\delta\cos(\theta) rn + \delta^2 n^2 }} dn
        \right|
    \end{align}
    Therefore we have
    \begin{align}
    \label{eq:kappaULA}&\kappa_N^2=
    {\big|y(\rbold,\bboldbarMR)\big|^2}/{\big|y(\rboldbar,\bboldbarMR)\big|^2}
        =
        \notag
        \\
        &
        \frac{
        \Bigg[
        \overbrace{
                \int_{-\frac{N}{2}}^{\frac{N}{2}} \frac{
        \cos(\eta_1 n + \eta_2 n^2) dn
        }
        {\sqrt{r^2-an + \delta^2 n^2 }
        }
        }^{A_N}
        \Bigg]^2
        +
        \Bigg[
         \overbrace{\int_{-\frac{N}{2}}^{\frac{N}{2}} \frac{
        \sin(\eta_1 n +\eta_2 n^2) dn
        }
        {\sqrt{r^2-an + \delta^2 n^2 }
        }}^{B_N} 
        \Bigg]^2
        }
        {
       \Bigg[
        \underbrace{\int_{-\frac{N}{2}}^{\frac{N}{2}} \frac{1}{\sqrt{\rbar^2-\overline{a} n + \delta^2 n^2 }} dn
        }_{C_N}\Bigg]^2
        }
        \end{align}
    where $a=2\delta\cos(\theta)r$ and $\overline{a}=2\delta\cos(\thetabar)\rbar$. It should be noted that $\rbold\neq \rboldbar$ results in $\eta_1\neq 0$ or $\eta_2 \neq 0$.
    First, we consider $\eta_1\neq 0$ and $\eta_2=0$. In this case, we can write $\lim_{N\rightarrow \infty} A_N\equiv A_{\infty}$ as follows:
    \begin{align}
        A_{\infty} &\approx  \int_{-\frac{N_0}{2}}^{\frac{N_0}{2}} \frac{
        \cos(\eta_1 n) dn
        }
        {\sqrt{r^2-an + \delta^2 n^2 }} + \frac{2}{\delta} \int_{\frac{N_0}{2}}^{\infty} \frac{\cos(\eta_1 n)}{n} dn
        \notag
        \\
        &=
        A_{N_0}-\frac{2}{\delta}\mathrm{ci}({\eta_1 N_0}/{2})
    \end{align}
    where $N_0/2$ is a threshold value, such that for $n>N_0/2$, the second-order term $\delta^2 n^2$ dominates the zero-order and first-order terms in the denominator of the integrand of $A_N$.
 The first term $A_{N0}$ is bounded since it corresponds to a proper integral, and the second term is also a bounded cosine integral (Equation 203(5) in \cite{gradshteyn2014table}). Similarly, for $B_N$ we can write
    \begin{align}
        B_{\infty}\approx B_{N_0}-\frac{2}{\delta}\mathrm{si}({\eta_1 N_0}/{2})
    \end{align}  
On the other hand, for large values of $N$, from \eqref{eq:kappaULA}, $C_{N}$ can be written as
\begin{align}
    C_{N}\approx C_{N'_0}+\frac{2}{\delta}\int_{\frac{N'_0}{2}}^{N}\frac{dn}{n}=C_{N'_0}-\frac{2}{\delta}\ln({2N}/{ N'_0}) 
\end{align}
Therefore, we conclude that
\begin{multline}
    \lim_{N\rightarrow \infty} \kappa_N \approx
    \\
    \lim_{N\rightarrow \infty}
    \frac{
    \sqrt{\left[A_{N_0}-\frac{2}{\delta}\mathrm{ci}(\frac{\eta_1 N_0}{2})\right]^2 + \left[ B_{N_0}-\frac{2}{\delta}\mathrm{si}(\frac{\eta_1 N_0}{2})\right]^2}
    }
    {C_{N'_0}-\frac{2}{\delta}\ln({2N}/{ N'_0})}=0
\end{multline}
Now we consider $\eta_2\neq 0$. Similarly, after doing some mathematical manipulations, we can show that
\begin{multline}
    \lim_{N\rightarrow \infty} \kappa_N \approx
    \\
    \lim_{N\rightarrow \infty}
    \frac{
    \sqrt{
    \left[A_{N_0}-\frac{1}{\delta}\mathrm{ci}(\frac{\eta_2 N_0^2}{4})\right]^2 + \left[ B_{N_0}-\frac{1}{\delta}\mathrm{si}(\frac{\eta_2 N_0^2}{4})\right]^2}
    }
    {C_{N'_0}-\frac{2}{\delta}\ln({2N}/{ N'_0})}=0
\end{multline}
\vspace{-10pt}
\section{Proof of Theorem   \ref{th:UPA_asymptotic_ratio}}
      Consider a $N$-element ($N=N_1N_2$) UPA on the $xz$ plane as shown in Fig. \ref{fig:system_model}-a with an interelement spacing of $\delta$, centered at the origin, where the central reference element corresponds to $n_1=0$ and $n_2=0$. 
    From \eqref{eq:rn_versus_r} the distance of the antenna element $(n_1,n_2)$ from points $\rbold$ and $\rboldbar$ are respectively written as $r_{n_1,n_2} =\sqrt{r^2-2ar n_1-2br n_2 + \delta^2 n_1^2 + \delta^2 n_2^2 }$ and $\overline{r}_{n_1,n_2} =\sqrt{\rbar^2-2\overline{a}\hspace{1pt}\rbar n_1-2\overline{b}\rbar n_2 + \delta^2 n_1^2 + \delta^2 n_2^2 }$ where $a=\delta\sin(\theta)\cos(\phi)$, $b=\delta\cos(\theta)$, $\overline{a}=\delta\sin(\thetabar)\cos(\overline{\phi})$, and $\bbar=\delta\cos(\thetabar)$.
    Considering the Taylor polynomials of $r_{n_1,n_2}$ and $\rbar_{n_1,n_2}$ 
    , we can write
    \begin{multline}
        \frac{2\pi}{\lambda}(r_{n_1,n_2}-\rbar_{n_1,n_2}) \approx \frac{2\pi}{\lambda}(r-\rbar) + 
        \\
        \underbrace{\zeta_{1} n_1 + \zeta_{2} n_2 + \zeta_3 n_1^2 + \zeta_4 n_2^2 + \zeta_5 n_1n_2}_{\widetilde{u}(n_1,n_2)}    
    \end{multline}
    in which $\zeta_1=\overline{a}-a$, $\zeta_2=\overline{b}-b$, $\zeta_3= \left(\frac{\delta^2 -a^2}{2r} - \frac{\delta^2 -\overline{a}^2}{2\rbar}\right)$, $\zeta_4= \left(\frac{\delta^2 -b^2}{2r}- \frac{\delta^2 -\overline{b}^2}{2\rbar}\right)$ and $\zeta_5= \left(\frac{\overline{a}\overline{b}}{\rbar^2} - \frac{ab}{r^2} \right)$. Similar to Theorem \ref{th:ULA_asymptotic_ratio} we have
    \begin{align}
        \lim_{N_1,N_2\rightarrow \infty} \kappa_{N_1,N_2}^2=
    \lim_{N_1,N_2\rightarrow \infty}
    \left[
        \left(\frac{A_{N_1,N_2}}{C_{N_1,N_2}}\right)^2+ \left(\frac{B_{N_1,N_2}}{C_{N_1,N_2}}\right)^2
        \right]
    \end{align}
    where 
    \begin{align}
        A_{N_1,N_2} &=  \int_{-\frac{N_1}{2}}^{\frac{N_1}{2}} \! \int_{-\frac{N_2}{2}}^{\frac{N_2}{2}} 
        \!
        \frac{
            \cos(\widetilde{u}(n_1,n_2)) dn_2 dn_1
        }
        {
            \sqrt{r^2-2ar n_1-2br n_2 + \delta^2 n_1^2 + \delta^2 n_2^2 }
        }
        \notag
        \\
        B_{N_1,N_2} &=  \int_{-\frac{N_1}{2}}^{\frac{N_1}{2}} \! \int_{-\frac{N_2}{2}}^{\frac{N_2}{2}} 
        \!
        \frac{
            \sin(\widetilde{u}(n_1,n_2)) dn_2 dn_1
        }
        {
            \sqrt{r^2-2ar n_1-2br n_2 + \delta^2 n_1^2 + \delta^2 n_2^2 }
        }
        \notag
        \\
        C_{N_1,N_2} &=  \int_{-\frac{N_1}{2}}^{\frac{N_1}{2}} \! \int_{-\frac{N_2}{2}}^{\frac{N_2}{2}} 
        \!
        \frac{
            dn_2 dn_1
        }
        {
            \sqrt{\rbar^2-2\overline{a}\rbar n_1-2\overline{b}\rbar n_2 + \delta^2 n_1^2 + \delta^2 n_2^2 }
        }
    \end{align}
    Similar to the proof of Theorem \ref{th:ULA_asymptotic_ratio}, it can be verified that $C_{N_1,N_2}$ is asymptotically divergent, while $A_{N_1,N_2}$ and $B_{N_1,N_2}$ are bounded values, provided that at least one of the coefficients $\zeta_1,...,\zeta_5$ be non-zero (which corresponds to $r\neq \rbar$, or $\theta\neq \thetabar$, or $\phi \neq \overline{\phi}$). 
    The details are omitted for conciseness.
\end{appendices}


\bibliographystyle{IEEEtran}

\bibliography{Mybib}

 \vspace{-10pt}
 
\begin{biography}[{\includegraphics[width=1in,height
=1.25in,clip,keepaspectratio]{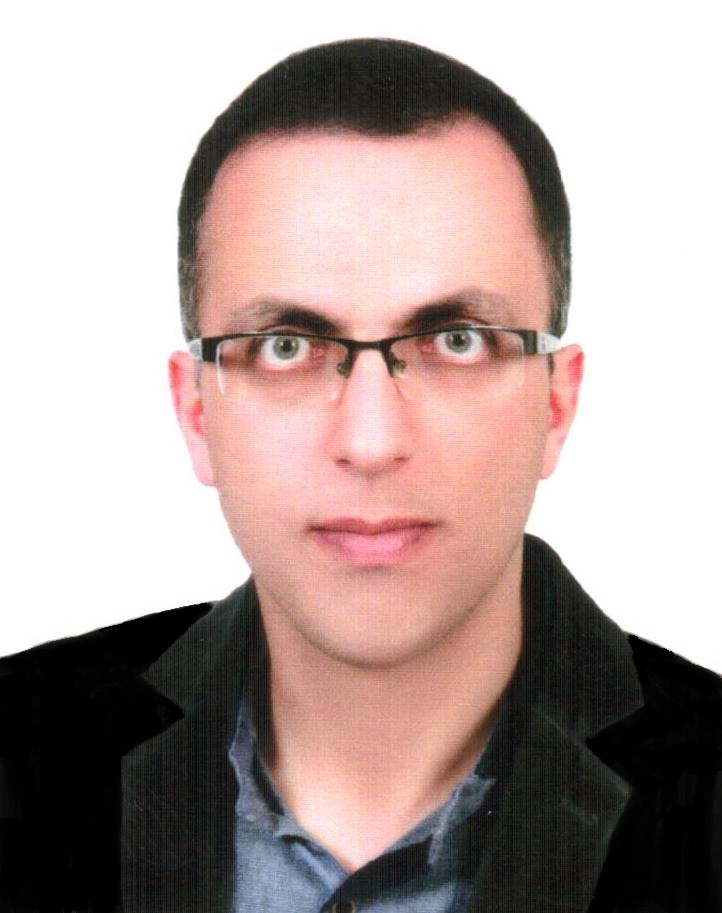}}]{Mehdi Monemi} (Member, IEEE)
		received the B.Sc., M.Sc., and Ph.D. degrees all in electrical and computer engineering from Shiraz University, Shiraz, Iran, and Tarbiat Modares University, Tehran, Iran, and Shiraz University, Shiraz, Iran in 2001, 2003 and 2014 respectively. After receiving his Ph.D., he worked as a project manager in several companies and was an assistant professor in the Department of Electrical Engineering, Salman Farsi University of Kazerun, Kazerun, Iran, from 2017 to May 2023. He was a visiting researcher in the Department of Electrical and Computer Engineering, York University, Toronto, Canada from June 2019 to September 2019. He is currently a Postdoc researcher with the Centre
for Wireless Communications (CWC), University of Oulu, Finland. His current research interests include resource allocation in 5G/6G networks, as well as the employment of machine learning algorithms in wireless networks.
	\end{biography}
 \vspace{-10pt}
\begin{biography}[{\includegraphics[width=1in,height
=1.25in,clip,keepaspectratio]{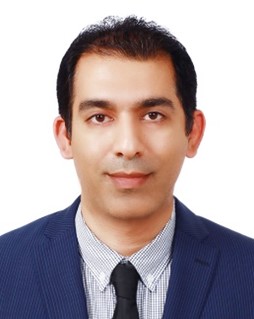}}]{Sirous Bahrami}
		received the B.S. degree from Isfahan University, Isfahan, Iran, in 2007 and the M.S. and Ph.D. degrees in electrical engineering from the Iran University of Science and Technology (IUST), Tehran, Iran, in 2009 and 2014, respectively.
From 2014 to 2020, he was with the Department of Electrical Engineering, Salman Farsi University of Kazerun, Kazerun, Iran. 
He joined the Pohang University of Science and Technology (POSTECH) as a Post-doctoral researcher in 2021. He is currently a Researcher Assistant Professor at POSTECH, Pohang, South Korea, developing multifeed and reconfigurable antennas for RF and mm-wave communication systems. His research interests include passive and active microwave devices, planar antennas, and monolithic microwave integrated circuits (MMICs).
Dr. Bahrami was a recipient of the best paper award 2024 Electromagnetic Measurement Technology South Korea, the best student paper award Finalist IEEE RFIC 2023, Top Reviewers award of IEEE TAP in 2024 and the best antenna measurement paper award finalist ISAP 2024.
	\end{biography}

\begin{biography}[{\includegraphics[width=1in,height=1.25in,clip,keepaspectratio]{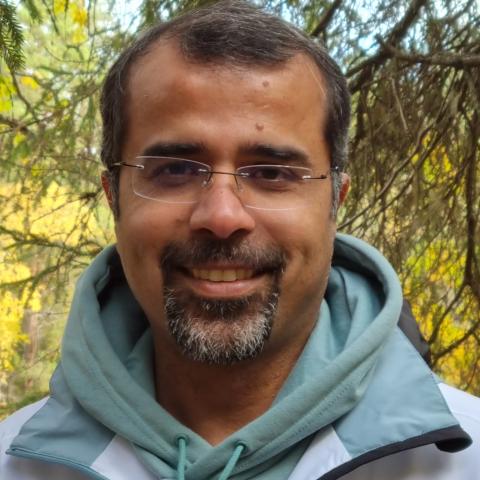}}]{Mehdi Rasti}
		(Senior Member, IEEE) received the B.Sc. degree in electrical engineering from Shiraz University, Shiraz, Iran, in 2001, and the M.Sc. and Ph.D. degrees from Tarbiat Modares University, Tehran, Iran, in 2003 and 2009, respectively. He is currently an Associate Professor with the Centre for Wireless Communications, University of Oulu, Finland. From 2012 to 2022, he was with the Department of Computer Engineering, Amirkabir University of Technology, Tehran. From February 2021 to January 2022, he was a Visiting Researcher with the Lappeenranta-Lahti University of Technology, Lappeenranta, Finland. From November 2007 to November 2008, he was a Visiting Researcher with the Wireless@KTH, Royal Institute of Technology, Stockholm, Sweden. 
  From June 2013 to August 2013, and from July 2014 to August 2014 he was a visiting researcher in the Department of Electrical and Computer Engineering, University of Manitoba, Winnipeg, MB, Canada. His current research interests include radio resource allocation in IoT, Beyond 5G and 6G wireless networks.
	\end{biography}

 \begin{biography}[{\includegraphics[width=1in,height=1.25in,clip,keepaspectratio]{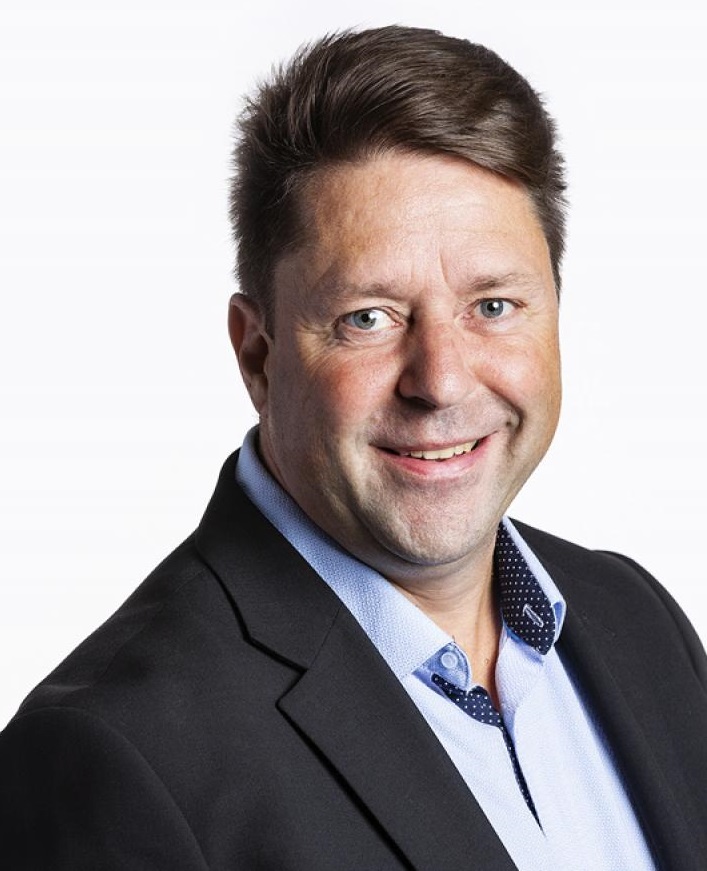}}]{Matti Latva-aho} (Fellow, IEEE) Matti Latva-aho (IEEE Fellow) received his M.Sc., Lic.Tech., and Dr.Tech. (Hons.) degrees in Electrical Engineering from the University of Oulu, Finland, in 1992, 1996, and 1998, respectively. From 1992 to 1993, he was a Research Engineer at Nokia Mobile Phones in Oulu, Finland, after which he joined the Centre for Wireless Communications (CWC) at the University of Oulu. Prof. Latva-aho served as Director of CWC from 1998 to 2006 and was Head of the Department of Communication Engineering until August 2014. He is currently a Professor of Wireless Communications at the University of Oulu and the Director of the National 6G Flagship Programme. He is also a Global Fellow at The University of Tokyo. Prof. Latva-aho has published over 500 conference and journal papers in the field of wireless communications. In 2015, he received the Nokia Foundation Award for his achievements in mobile communications research.
	\end{biography}

\end{document}